%% file: PANIGRAHAM.tex
\documentclass{article}
\usepackage{arxiv}
\usepackage{cite}
\usepackage{amsmath,amssymb,amsfonts}
\usepackage{graphicx}
\usepackage{textcomp}
\usepackage{xcolor}

\usepackage{pdfpages}

\usepackage{booktabs} 
\usepackage{xspace}
\usepackage{mathtools}
\usepackage{amsmath,amssymb}
\usepackage{enumitem}
\usepackage{verbatim}
\usepackage{ragged2e}
\usepackage{algorithm}
\usepackage{algorithmicx}
\usepackage{algcompatible}
\usepackage{algpseudocode}
\algrenewcommand\algorithmicindent{0.7em}
\usepackage{graphicx} 
\usepackage{color}
\usepackage{caption}
\usepackage{ragged2e}
\usepackage{tabularx}
\usepackage{pgfplotstable}
\pgfplotsset{compat=1.11,
        /pgfplots/ybar legend/.style={
        /pgfplots/legend image code/.code={%
        \draw[##1,/tikz/.cd,bar width=3pt,yshift=-0.2em,bar shift=0pt]
                plot coordinates {(0cm,0.8em)};},
},
}

\usepackage{calligra}
\usepackage{calrsfs}
\DeclareMathAlphabet{\mathcalligra}{T1}{calligra}{m}{n}
\algrenewcommand\alglinenumber[1]{\tiny #1:}
\usepackage{amsmath}
\usepackage{mathtools}
\usepackage{amsmath,amssymb,latexsym} 
\usepackage{float}
\usepackage{breqn}
\usepackage{xcolor}

\usepackage{amsthm}

 \usepackage{relsize}
\usepackage[justification=centering]{caption}

\pgfplotsset{every tick label/.append style={font=\Large}}
\usepackage[justification=centering,font=footnotesize]{caption}
\usepackage{caption}
\usepackage{subcaption}
\usepackage{pgf,tikz}
\usetikzlibrary{matrix}
\usepackage{pgfplots}
\pgfplotsset{width=7cm,compat=1.8}
\usetikzlibrary{arrows}
\usetikzlibrary{shadows}
\usepackage[utf8]{inputenc}
\usepackage{xcolor, soul}
\usepackage{multicol}
\usepackage{enumitem}
\usepackage{soul}
\usepackage{mathtools}
\usepackage{breqn}
\usepackage{url}
\usepackage{listings}
\usepackage{todonotes}
\definecolor{bblue}{HTML}{4F81BD}
\definecolor{rred}{HTML}{C0504D}
\definecolor{ggreen}{HTML}{9BBB59}
\definecolor{ppurple}{HTML}{9F4C7C}
\title{Dynamic Graph Operations: A Consistent Non-blocking Approach}
\author{
	Bapi Chatterjee\thanks{Supported by the European Union's Horizon 2020 
		research and innovation programme under the Marie Skodowska-Curie grant 
		agreement No. 754411 (ISTPlus).}\\
	Institute of Science and Technology\\ 
	Austria \\
	\texttt{bapi.chatterjee@ist.ac.at}\\
	\And
	Sathya Peri\thanks{Partially supported by a research grant from Intel, USA.} \\
	Department of Computer Science \& Engineering\\
	Indian Institute of Technology Hyderbad, India\\
	\texttt{sathya\_p@iith.ac.in}
	\And
	Muktikanta Sa$^\dagger$ \\
	Department of Computer Science \& Engineering\\
	Indian Institute of Technology Hyderbad, India\\
	\texttt{cs15resch11012@iith.ac.in}
}
\date{}
\input{macros}

\begin{document}

\maketitle

\begin{abstract}
\input{abstract}
\\
\textbf{keywords:}{concurrent data structure, lock-free, obstruction-free,  directed graph, breadth-first-search, single-source-shortest-path,  betweenness 
centrality.}
\end{abstract}

\section{Introduction} \label{sec:intro}
\input{intro}
\section{\Nbk Graph Data Structure}\label{sec:graph-ds}
\input{graph-ds}  
\section{PANIGRAHAM Framework}\label{sec:graph-algo}
\input{graph-algo}

\section{The Graph Query Operations}\label{sec:applications}
\input{applications}   
\section{Experiments}
\label{sec:results}
 \input{results.tex}

\bibliographystyle{plain}
\bibliography{biblio}

 \newpage
 \appendix
 \input{appendix}
\end{document}

%% file: macros.tex
\newcommand{\punt}[1]{}
\newcommand{\cmnt}[1]{}
\algnewcommand{\IIf}[1]{\State\algorithmicif\ #1\ \algorithmicthen}
\algnewcommand{\EndIIf}{\unskip\ \algorithmicend\ \algorithmicif}

\newcommand{\lble} {linearizable\xspace}
\newcommand{\lbty} {linearizability\xspace}
\newcommand{\rbty} {reachability\xspace}

\definecolor{darkblue}{rgb}{0.0, 0.0, 0.55}
\newcommand{\linecomment}[1]{{\scriptsize \textcolor{darkblue}{#1}}}
\newtheorem{theorem}{Theorem}

\newcounter{history}

\newcommand{\secref}[1]{Section~\ref{sec:#1}}
\newcommand{\figref}[1]{Figure~\ref{fig:#1}}

\newcommand{\lineref}[1]{Line~\ref{lin:#1}}
\newcommand{\linref}[1]{Line~\ref{lin:#1}}

\def\Statenolinnum#1{{\def\alglinenumber##1{}\State #1}\addtocounter{ALG@line}{-1}}

\newcommand{\ignore}[1]{}

\newcommand{\myparagraph}[1]{\noindent\textbf{#1.}}
\algdef{SE}[DOWHILE]{Do}{doWhile}{\algorithmicdo}[1]{\algorithmicwhile\ #1}%
\newcommand{\op} {operation\xspace}
\newcommand{\mth} {method\xspace}
\newcommand{\cc} {correctness-criterion\xspace}

\newcommand{\lp} {LP\xspace}

\newcommand{\tru} {\texttt{true}\xspace}
\newcommand{\cntu} {\texttt{continue}\xspace}
\newcommand{\fal} {\texttt{false}\xspace}
\newcommand{\nul} {\texttt{NULL}\xspace}
\newcommand{\brk} {\texttt{break}\xspace}

\newcommand{\vnodes} {{\tt VNodes}\xspace}
\newcommand{\enodes} {{\tt ENodes}\xspace}
\newcommand{\enode}{{\tt ENode}\xspace}
\newcommand{\vnode}{{\tt VNode}\xspace}
\newcommand{\bfsnode}{{\tt BFSNode}\xspace}
\newcommand{\bfslistnode}{{\tt BFSListNode}\xspace}

\newcommand{\spnode}{{\tt SPNode}\xspace}
\newcommand{\pnode}{{\tt PNode}\xspace}
\newcommand{\bfstree}{{\tt BFS\text{-}tree}\xspace}

\newcommand{\sptree}{{\tt SP\text{-}tree}\xspace}
\newcommand{\vlist} {vertex-list\xspace}
\newcommand{\elist} {edge-list\xspace}
\newcommand{\elists} {edge-lists\xspace}

\newcommand{\vh}{\texttt{vh}\xspace}
\newcommand{\vt}{\texttt{vt}\xspace}
\newcommand{\eh}{\texttt{eh}\xspace}

\newcommand{\addv}{\textsc{PutV}\xspace}
\newcommand{\remv}{\textsc{RemV}\xspace}
\newcommand{\adde}{\textsc{PutE}\xspace}
\newcommand{\reme}{\textsc{RemE}\xspace}
\newcommand{\conv}{\textsc{GetV}\xspace}
\newcommand{\cone}{\textsc{GetE}\xspace}

\newcommand{\opn}{\textsc{Op}\xspace}

\newcommand{\locvplus}{\textsc{locV}\xspace}
\newcommand{\loceplus}{\textsc{locE}\xspace}

\newcommand{\createe} {\textsc{CEnode}\xspace}
\newcommand{\createv}{\textsc{CVnode}\xspace}
\newcommand{\convplus}{\textsc{ConVPlus}\xspace}

\newcommand{\add}{\textsc{Add}\xspace}
\newcommand{\addsp}{\textsc{AddToSPTree}\xspace}
\newcommand{\updatesp}{\textsc{UpdateSPTree}\xspace}

\newcommand{\rem}{\textsc{Remove}\xspace}
\newcommand{\con}{\textsc{Contains}\xspace}

\newcommand{\fadd}{FetchAndAdd\xspace}

\newcommand{\rlx}{\textsc{Rlxd}\xspace}

\newcommand{\checkvisited}{\textsc{ChkVisit}\xspace}
\newcommand{\visitedarray}{\texttt{VisA}\xspace}

\newcommand{\of}{obstruction-free\xspace}
\newcommand{\Nbk}{Non-blocking\xspace}
\newcommand{\nbk}{non-blocking\xspace}

\newcommand{\vcs}{vertices\xspace}

\newcommand{\cas}{compare-and-swap\xspace}
\newcommand{\CAS}{\texttt{CAS}\xspace}
\newcommand{\faa}{fetch-and-add\xspace}
\newcommand{\FAA}{\texttt{FAA}\xspace}

\newcommand{\ds}{data structure\xspace}
\newcommand{\cds}{concurrent data structure\xspace}

\newcommand{\lf}{lock-free\xspace}
\newcommand{\wf}{wait-free\xspace}

\newcommand{\node}{node}
\newcommand{\Node}{Node}

\newcommand{\head}{{\tt Head}\xspace}

\newcommand{\Init}{\textsc{Init\xspace}}

\newcommand{\enext}{{\tt enxt}\xspace}
\newcommand{\vnext}{{\tt vnxt}\xspace}
\newcommand{\bnext}{{\tt nxt}\xspace}
\newcommand{\eleft}{{\tt el}\xspace}
\newcommand{\eright}{{\tt er}\xspace}

\newcommand{\pointv}{{\tt ptv}\xspace}
\newcommand{\ecount}{{\tt ecnt}\xspace}
\newcommand{\lecount}{{\tt ecnt}\xspace}

\newcommand{\tcount}{{\tt cnt}\xspace}

\newcommand{\eweight}{{\tt w}\xspace}
\newcommand{\opitem}{{\tt oi}\xspace}
\newcommand{\opstruct}{{\tt OpItem}\xspace}

\newcommand{\distarray}{{\tt DistA}\xspace}

\newcommand{\help}{Help\xspace}

\newcommand{\eop}[2]{$\langle$#1,#2$\rangle$\xspace}

\newcommand{\ventp}{{\tt VERTEX OR EDGE NOT PRESENT}\xspace}

\newcommand{\ep}{{\tt EDGE PRESENT}\xspace}

\newcommand{\scan}{\textsc{Scan}\xspace}
\newcommand{\collect}{\textsc{Collect}\xspace}

\newcommand{\comparetree}{\textsc{CmpTree}\xspace}
\newcommand{\spcomparetree}{\textsc{SPCmpTree}\xspace}

\newcommand{\comparetreegraph}{\textsc{CmpTGph}\xspace}

\newcommand{\isMarked}{\textsc{isMrkd}\xspace}
\newcommand{\MarkedRef}{\textsc{Mrk}\xspace}
\newcommand{\unMarkedRef}{\textsc{UnMrk}\xspace}

\newcommand{\putv}{\textsc{PutV}\xspace}

\newcommand{\pute}{\textsc{PutE}\xspace}

\newcommand{\Put}{\textsc{Put}\xspace}
\newcommand{\Get}{\textsc{Get}\xspace}

\newcommand{\sptclt}{\textsc{SPTClt}\xspace}
\newcommand{\bfstclt}{\textsc{BFSTClt}\xspace}

\newcommand{\bctclt}{\textsc{BCTClt}\xspace}

\newcommand{\diametertclt}{\textsc{DiaTClt}\xspace}

\newcommand{\checknegcycle}{\textsc{CheckNegCycle}\xspace}

\newcommand{\getsp}{\textsc{SSSP}\xspace}
\newcommand{\getbfs}{\textsc{BFS}\xspace}

\newcommand{\getbc}{\textsc{BC}\xspace}

\newcommand{\getdia}{\textsc{GD}\xspace}

\newcommand{\spscan}{\textsc{SPScan}\xspace}

\newcommand{\bfsscan}{\textsc{BFSScan}\xspace}

\newcommand{\bcscan}{\textsc{BCScan}\xspace}

\newcommand{\diascan}{\textsc{DiaScan}\xspace}

\newcommand{\BFS}{Breadth-first search\xspace}
\newcommand{\bfs}{breadth-first search\xspace}

\newcommand{\sssp}{single-source shortest-path\xspace}

\newcommand{\bc}{betweenness centrality\xspace}

\newcommand{\idx}{\texttt{index}\xspace}

\newcommand{\sizee}{\texttt{size}\xspace}
\newcommand{\sigmaa}{\texttt{sigmaA}\xspace}
\newcommand{\betweenness}{\texttt{CbA}\xspace}
\newcommand{\deltaa}{\texttt{deltaA}\xspace}
\newcommand{\predlist}{\texttt{PredlistA}\xspace}

\newcommand{\level}{\texttt{level}\xspace}
\newcommand{\maxlevel}{\texttt{maxLevel}\xspace}

\newcommand{\Ligra}{Ligra\xspace}

\newcommand{\panighm}{PANIGRAHAM\xspace}

\newcommand{\vkey}{{\tt v}\xspace}
\newcommand{\ekey}{{\tt e}\xspace}
\newcommand{\key}{{\tt key}\xspace}

\newcommand{\getgphalgo}{\textsc{Op}\xspace}
\newcommand{\treeclt}{\textsc{TreeCollect}\xspace}
\newcommand{\createtreenode}{\textsc{CTNode}\xspace}
\newcommand{\treenode}{{\tt SNode}\xspace}

\newcommand{\initbucket}{\textsc{initBkt}\xspace}
\newcommand{\getresponse}{\textsc{GetResponse}\xspace}
\newcommand{\hasmember}{\textsc{HasMember}\xspace}
\newcommand{\freeze}{\textsc{Freeze}\xspace}
\newcommand{\invoke}{\textsc{Invoke}\xspace}
\newcommand{\apply}{\textsc{Apply}\xspace}
\newcommand{\resize}{\textsc{Resize}\xspace}

\newcommand{\Mod}{\textbf{mod}\xspace}
\newcommand{\Hsize}{{\tt size}\xspace}
\newcommand{\hpred}{{\tt pred}\xspace}
\newcommand{\set}{{\tt set}\xspace}
\newcommand{\ok}{{\tt ok}\xspace}
\newcommand{\optype}{{\tt optype}\xspace}
\newcommand{\done}{{\tt done}\xspace}
\newcommand{\resp}{{\tt resp}\xspace}
\newcommand{\bucket}{{\tt bucket}\xspace}
\newcommand{\buckets}{{\tt buckets}\xspace}
\newcommand{\grow}{{\tt grow}\xspace}

\newcommand{\fset}{{\tt FSet}\xspace}
\newcommand{\fsetnode}{{\tt FSetNode}\xspace}
\newcommand{\fsetop}{{\tt FSetOp}\xspace}
\newcommand{\hnode}{{\tt HNode}\xspace}
\newcommand{\hnodes}{{\tt HNodes}\xspace}

\newcommand{\operation}{{\tt Operation}\xspace}
\newcommand{\oper}{{\tt op}\xspace}
\newcommand{\lft}{{\tt left}\xspace}
\newcommand{\rght}{{\tt right}\xspace}
\newcommand{\statee}{{\tt state}\xspace}
\newcommand{\childcasop}{{\tt ChildCASOp}\xspace}
\newcommand{\isleft}{{\tt ifLeft}\xspace}
\newcommand{\expected}{{\tt expected}\xspace}
\newcommand{\updt}{{\tt update}\xspace}
\newcommand{\expert}{{\tt expert}\xspace}
\newcommand{\relocateop}{{\tt RelocateOp}\xspace}
\newcommand{\dest}{{\tt dest}\xspace}
\newcommand{\destop}{{\tt destOp}\xspace}
\newcommand{\removekey}{{\tt removeKey}\xspace}
\newcommand{\replacekey}{{\tt replaceKey}\xspace}
\newcommand{\Root}{{\tt root}\xspace}
\newcommand{\VCAS}{{\tt VCAS}\xspace}

\newcommand{\find}{\textsc{Find}\xspace}

\newcommand{\flag}{\textsc{Flag}\xspace}
\newcommand{\helpchildcas}{\textsc{HelpChildCAS}\xspace}
\newcommand{\helprelocate}{\textsc{HelpRelocate}\xspace}
\newcommand{\isnull}{\textsc{IsNull}\xspace}
\newcommand{\getflag}{\textsc{GetFlag}\xspace}
\newcommand{\unflag}{\textsc{UnFlag}\xspace}

\newcommand{\helpmarked}{\textsc{HelpMarked}\xspace}
\newcommand{\ischildcas}{\textsc{IsChildCAS }\xspace}
\newcommand{\setnull}{\textsc{SetNull}\xspace}
\newcommand{\isrelocate}{\textsc{IsRelocate}\xspace}
\newcommand{\gettid}{\textsc{GetThId}\xspace}

\newcommand{\found}{{\tt FOUND}\xspace}
\newcommand{\none}{{\tt NONE}\xspace}
\newcommand{\notfoundl}{{\tt NOTFOUND\_L}\xspace}
\newcommand{\notfoundr}{{\tt NOTFOUND\_R}\xspace}
\newcommand{\marked}{{\tt MARKED}\xspace}
\newcommand{\childcas}{{\tt CHILDCAS}\xspace}
\newcommand{\abort}{{\tt ABORT}\xspace}
\newcommand{\relocate}{{\tt RELOCATE}\xspace}
\newcommand{\ongoing}{{\tt ONGOING}\xspace}
\newcommand{\caschild}{{\tt CASCHILD}\xspace}
\newcommand{\successful}{{\tt SUCCESSFUL}\xspace}
\newcommand{\failed}{{\tt FAILED}\xspace}

\newcommand{\hashadd}{\textsc{HashAdd}\xspace}
\newcommand{\hashrem}{\textsc{HashRem}\xspace}
\newcommand{\hashcon}{\textsc{HashCon}\xspace}

\newcommand{\bstcon}{\textsc{BSTCon}\xspace}

%% file: abstract.tex
Graph algorithms enormously contribute to the domains such as 
blockchains, social networks, biological networks, telecommunication networks, 
and several others. The ever-increasing demand of data-volume as well as speed 
of such applications have essentially transported these applications from their 
comfort zone: static setting, to a challenging territory of dynamic updates. 
At the same time, mainstreaming of multi-core processors have entailed that the 
dynamic applications should be able to exploit concurrency as soon as 
parallelization gets inhibited. Thus, the design and 
implementation of efficient concurrent dynamic graph algorithms have become 
significant. 

This paper reports a novel library of concurrent shared-memory algorithms for 
breadth-first search (BFS), single-source shortest-path (SSSP), and betweenness 
centrality (BC) in a dynamic graph. The presented 
algorithms are provably  non-blocking and linearizable. We extensively evaluate 
C++ implementations of the algorithms through several micro-benchmarks. The 
experimental results demonstrate the scalability with the number of threads. 
Our experiments also highlight the limitations of static graph analytics 
methods in dynamic setting. 

%% file: intro.tex
A graph represents the pairwise relationships between objects or entities that underlie the complex frameworks such as blockchains \cite{AnjanaKPRS19}, social networks \cite{CataneseMFFP11}, 
semantic-web \cite{bhatia2018understanding}, biological networks \cite{bti167}, and many others. Often these applications are implemented on \textit{dynamic} graphs: they undergo changes like addition and removal of vertices and/or edges\cite{Demetrescu+:DynGraph::book:2004} over time. 
For example, consider the computation of shortest path or centrality between nodes in a real-time dynamically changing social network as highlighted in \cite{KasCC13}. 
Such settings are challenging and approaches such 
as incremental computation \cite{KasCC13} or streaming framework, where a graph operation is performed over a static temporal snapshot of the data structure, e.g. Kineograph \cite{ChengHKMWWYZZC12}, GraphTau \cite{IyerLDS16}, are currently adopted. However, application of concurrency in dynamic graph algorithms is largely unexplored where dynamic dataset-updates severely hinder parallel operation-processing designed for static graphs. 

%
With the rise of multi-core computers around a decade back, \emph{\cds{s}} have become popular, for they are able to harness the power of multiple cores effectively. Several \cds{s} have been developed: stacks~\cite{Hendler+:LFStack:SPAA:2004}, queues~\cite{Barnes:LFDS:SPAA:1993, 
Herlihy+:OFDQue:icdcs:2003, Kogan+:WFQue:ppopp:2011, 
Stellwag+:WFDque:sies:2009, Shavit+:LFQ:DC:2008}, 
linked-lists~\cite{Harris:NBList:disc:2001,Valois:LFList:podc:1995,Zhang+:NBUnList:disc:2013,Heller+:LazyList:PPL:2007,Timnat:WFLis:opodis:2012,chatterjee2016help},
 hash tables~\cite{Michael:LFHashList:spaa:2002,Liu+:LFHash:PODC:2014}, binary search 
trees~\cite{EllenFRB10,Natarajan+:LFBST:ppopp:2014,Chatterjee:+LFBST:PODC:2014, BrownER14,Ramachandran+:LFIBST:ICDCN:2015,chatterjee2016help}, etc. On concurrent graphs, Kallimanis et al. \cite{Kallimanis+:WFGraph:opodis:2015} presented dynamic traversals and Chatterjee et al. \cite{Chatterjee+:NbGraph:ICDCN-19} presented \rbty queries. A data structure allowing concurrent update operations is readily suitable to accommodate dynamic updates. The array of useful graph queries, in particular, those applied in graph analytics, span much beyond \rbty queries. For example, single-source-shortest-path (SSSP) queries appertain to link-prediction in social networks~\cite{bhatia2018understanding}. Similarly, betweenness centrality algorithm finds applications in chemical network analysis~\cite{zhao2015application}. The aforementioned queries inherently scan through (almost) the entire graph. In a dynamic setting, a concurrent update of a vertex or an edge can potentially render the output of such queries inconsistent. 

To elucidate, consider computing the shortest path between two vertices. It requires exploring all possible paths between them followed by returning the set of edges that make the shortest path. It is easy to see that an addition of an edge to another path can make it shorter than the one returned, and similarly, a removal of an edge (from it) could turn it no longer the shortest. Imagine the addition and removal to be concurrent with the query, which can certainly benefit the application. Clearly, the return of the query can be inconsistent with the latest state of the graph.
\ignore{
Naturally, the lifetime of an addition or a removal operation is much shorter than that of a shortest path query and therefore the return of the query would be inconsistent with the state of the graph. 
}

A commonly accepted \cc for \cds{s} is \emph{\lbty} \cite{Herlihy+:lbty:TPLS:1990}, which intuitively infers that the observed output of a concurrent execution of a set of operations should be as if the 
operations were executed in a certain sequential order. Separating a graph query from concurrent updates by way of locking the shared vertices and edges can achieve \lbty. However, locking the portion of the graph, that requires access by a query, which often could very well be its entirety, would obstruct a large number of concurrent fast updates. Even an effortful interleaving of the query- and update-locks 
at a finer granularity does not protect against pitfalls such as deadlock, convoying, etc. A more attractive option is to implement \emph{\nbk} progress, which ensures that some non-faulty (non-crashing) threads complete their operations in a finite number of steps. Unsurprisingly, \nbk \lble design of queries that synchronize with concurrent updates in a dynamic graph is exacting.

In this paper, we describe the design and implementation of a graph data structure library (1) that supports three commonly useful operations -- breadth-first search (BFS), single-source shortest-path (SSSP), and betweenness centrality (BC) (2) that supports dynamic updates of edges and vertices concurrent with the operations and (3) ensures \lbty and \nbk progress. We call it PANIGRAHAM\footnote{Panigraham is the Sanskrit translation of Marriage, which undoubtedly is a prominent phenomenon in our lives resulting in networks represented by graphs.}: \textbf{P}r\textbf{a}ctical \textbf{N}on-block\textbf{i}ng \textbf{Gra}ph Algorit\textbf{hm}s. We formally introduce these graph operations below.

A \textit{graph} is represented as a pair $G = (V,E)$, where $V$ is the set of \textit{\vcs} and $E$ is the set of \textit{edges}. An edge $e \in E$, $e\coloneqq(u,v)$ represents a pair of vertices $u,v \in V$. In a \textit{directed} graph\footnote{In this paper we confine the scope of discussion to directed graphs only.} $e\coloneqq(u,v)$ is an ordered pair, thus has an associated direction: \textit{emanating (outgoing)} from $u$ and \textit{terminating (incoming)} at $v$. We denote the set of outgoing edges from $v$ by $E_v$. Thus, $\cup_{v\in V}E_v = E$. Each edge $e \in E$ has a \textit{weight} $w_e$. A node $v\in V$ is said \textit{reachable} from $u\in V$: $v \hookleftarrow u$ if there are consecutive edges $\{e_1,e_2,\ldots,e_n\}\subseteq E$ such that $e_1$ emanates from $u$ and $e_n$ terminates at $v$.

\begin{enumerate}[leftmargin=5.5mm]
	\item \textbf{BFS}: Given a query vertex $v \in V$, output each vertex $u \in V-v$ reachable from $v$. The collection of vertices happens in a \textit{BFS order}: those at a distance $d_1$ from $v$ is collected before those at a distance $d_2>d_1$.
	\item \textbf{SSSP}: Given a vertex $v \in V$, find a shortest path with respect to total edge-weight from 
	$v$ to every other vertex $u \in V-v$. Note that, given a pair of nodes 
	$u,v \in V$, the shortest path between $u$ and $v$ may not be unique.
	\item \textbf{BC}: Given a vertex $v\in V$, compute $BC(v) = \sum_{s, t \in V}\frac{\sigma(s, t|v)}{\sigma(s,t)}$, where $\sigma(s,t)$ is the number of shortest paths between vertices $s,t \in V$ and $\sigma(s, t|v)$ is that passing through $v$. $BC(v)$ indicates the prominence of $v$ in $V$ and finds several applications where influence of an entity in a network is to be measured.
\end{enumerate} 

\subsection*{Overview}
In a nutshell, we implement a concurrent non-blocking dynamic  directed graph data structure as a composition of lock-free sets: a lock-free hash-table and multiple lock-free BSTs. The set of 
outgoing edges $E_v$ from a vertex $v \in V$ is implemented by a BST, whereas, $v$ itself is a node of the hash-table. Addition/removal of a vertex amounts to the same of a node in the lock-free hash-table, whereas, addition/removal of an edge translates to the same operation in a lock-free BST. Although lock-free progress is composable \cite{Dang2011}, thereby ensuring lock-free updates in the graph, however, optimizing these operations are nontrivial. The operations -- BFS, SSSP, BC -- are implemented by specialized partial snapshots of the composite data structure. In a dynamic concurrent 
non-blocking setting, we apply multi-scan/validate \cite{AfekADGMS93} to ensure the \lbty of a partial snapshot. We prove that these operations are \textit{\nbk}. The empirical results show the effectiveness of our algorithms.

\subsection*{Contributions and paper summary}
\begin{itemize}[leftmargin=5.5mm]
	\item First, we introduce the ADT and describe the \nbk design of 
	directed graph data structure as a composition of lock-free sets. 
	(\secref{graph-ds})
	\item After that, we describe our novel framework as an interface operation with its correctness and progress guarantee (\secref{graph-algo}) followed by the detailed concurrent implementation of \getbfs, \getsp and \getbc. (Section \ref{sec:applications})
	\item We present an experimental evaluation of the data structure. A novel feature of our experiments  is comparison of the concurrent data structure against a parallel graph operations library Ligra~\cite{Shun+:ligra:ppopp:2013}. 
	Our experimental observations demonstrate the power of concurrency for dynamic updates in an application. Utilizing the parallel compute resources -- 56 threads -- in a standard multi-core machine, our implementation performs up to 10x better than Ligra for BFS, SSSP and BC algorithms in some cases. (\secref{results}) 
\end{itemize}

\subsection*{Related work}
The libraries of parallel 
implementation of graph operations are abundant in literature. A relevant survey can be found in~\cite{batarfi2015large}. To mention a few well-known ones: PowerGraph~\cite{gonzalez2012powergraph} Galois 
\cite{Kulkarni+:Galois:PLDI:2007}, Ligra~\cite{Shun+:ligra:ppopp:2013}, and the extensions thereof such as Ligra+~\cite{Shun+:ligra+:dcc:2015}. However, they primarily focus on static queries and natively do not allow updates to the data structure, let alone concurrency. The implementations on GPUs \cite{HongOO11} adapt the existing parallel algorithms to the enhanced available parallelism therein, whereas, on the streaming frameworks the same algorithms are applied on static snapshots \cite{IyerLDS16,ChengHKMWWYZZC12}. A couple of exceptions though: Stinger \cite{Ediger+:Stinger:hpec:2012} and Congra~\cite{pan2017congra} support concurrent updates, however, they steer away from discussing the main challenges accompanying concurrency -- progress guarantee and correctness. To our knowledge, the present work is the first in this direction. 

In the field of concurrent data structures, only a couple of previous attempts are known. Kallimanis et al. \cite{Kallimanis+:WFGraph:opodis:2015} presented wait-free graph with dynamic traversals. Their design is based on adjacency matrix, can not support an unbounded graph and has no known implementation. Chatterjee et al. \cite{Chatterjee+:NbGraph:ICDCN-19} presented a lock-free graph designed on lock-free linked-lists that support \rbty queries. They have limited implementation results. Thus, in the field of concurrent data structures, this is the first attempt to implement SSSP and BC, while ensuring \lbty and non-blocking progress.

\myparagraph{Remark} Before moving to technical details, we pause to 
emphasize our aim/technique in contrast to existing methods.
\begin{itemize}[leftmargin=5.5mm]
	\item We target more applied ``non-local'' operations in graph that are 
	used in the domains such as analytics. In contrast, local operations, for instance, vertex/edge coloring, which are highly interesting in their own merits, are not in the scope of this work.
	\item In general, non-blocking data structures attempt to ensure 
	lock-freedom which packs benefits from both worlds: performance and 
	progress guarantee, however, ensuring lock-freedom in snapshot-like 
	operations, such as ours, is particularly costly, and would easily lose out 
	to highly optimized parallel static graph queries. Therefore, in the spirit 
	of exploring advantages of concurrency in dynamic settings, we worked with 
	obstruction-freedom, which favors performance while still ensuring 
	practical progress guarantee. 
	\item Though we compare against a high performance graph analytics 
	framework, it is important to mention that the individual operations in our 
	work are not even ``inline'' parallelized. Essentially, our experiments 
	display the collective strength of ``sequential'' operations, which satisfy 
	\lbty while interleaving concurrently, in a 
	dynamic setting of large recurrence.
\end{itemize}

%% file: graph-ds.tex
Our discussion uses a standard shared-memory model~\cite{Chatterjee+:NbGraph:ICDCN-19} that supports atomic \texttt{read}, \texttt{write}, \texttt{\faa} (\FAA)  and \texttt{\cas} (\CAS) instructions.
\subsection*{The Abstract Data Type (ADT)}
Consider a weighted directed graph $G = (V, E)$ as defined before. A vertex  $v\in V$ has an immutable unique \textit{key} drawn from a totally ordered universe. For brevity, we denote a vertex with key \vkey: $v(\vkey)$ by \vkey itself. Extending on the notations used in Section \ref{sec:intro}, we denote a directed edge  with weight $w$ from the vertex $\vkey_1$ to $\vkey_2$ as $(\vkey_1, \vkey_2| 
w)\in E$. We consider an ADT $\mathcal{A}$ as a set of the operations $\mathcal{A}$:

\ignore{
=\{\addv(\vkey), \remv(\vkey), \conv(\vkey), \adde(\vkey_1, \vkey_2| w), \reme(\vkey_1, \vkey_2), \cone(\vkey_1,\vkey_2),\linebreak\getbfs(\vkey), \getsp(\vkey), \getbc(\vkey)\}$ on $G$. The detailed definition of $\mathcal{A}$ is available in \cite{CPS20}.
}

A precondition for $(\vkey_1, \vkey_2| w)\in E$ is $\vkey_1\in V \land \vkey_2 \in V$. 
\begin{enumerate}[leftmargin=5.5mm]
	\item A $\addv(\vkey)$ updates $V$ to $V\cup\vkey$ and returns \tru if $\vkey \notin V$, otherwise it returns \fal without any update. 
	\item A $\remv(\vkey)$ updates $V$ to $V-\vkey$ and returns \tru if $v(\vkey) \in V$, otherwise it returns \fal without any update.
	\item A $\conv(\vkey)$ returns \tru if $\vkey \in V$, and \fal if $\vkey \notin V$.
	\item A $\adde(\vkey_1, \vkey_2| w)$ 
	\begin{enumerate}[label=(\alph*),leftmargin=\parindent,align=left,labelwidth=\parindent,labelsep=1pt,nosep,topsep=0pt,itemsep=0pt]
		\item updates $E$ to $E \cup (\vkey_1, \vkey_2| w)$ and returns \eop{\tru}{$\infty$} if $\vkey_1 \in V \land \vkey_2 \in V \land (\vkey_1, \vkey_2|\cdot)\notin E$, 
		\item updates $E$ to $E -(\vkey_1, \vkey_2| z) \cup (\vkey_1, \vkey_2| w)$; returns \eop{\tru}{$z$} if $(\vkey_1, \vkey_2|z)\in E$, 
		\item returns \eop{\fal}{$w$} if $(\vkey_1, \vkey_2|w)\in E$ without updates,
		\item returns \eop{\fal}{$\infty$} if $\vkey_1 \notin V \lor \vkey_2 \notin V$ without updates.
	\end{enumerate}	
	\item A $\reme(\vkey_1, \vkey_2)$ updates $E$ to $E - (\vkey_1, \vkey_2| w)$ and returns \eop{\tru}{$w$} if $(\vkey_1, \vkey_2|w)\in E$, otherwise it returns \eop{\fal}{$\infty$} without any update.
	\item A $\cone(\vkey_1,\vkey_2)$ returns \eop{\tru}{$w$} if $(\vkey_1, \vkey_2|w)\in E$, otherwise it returns \eop{\fal}{$\infty$}.
	\item A $\getbfs(\vkey)$, if $\vkey \in V$, returns a sequence of vertices reachable from \vkey arranged in a BFS order as defined before. If $\vkey \notin V$ or $\nexists \vkey' \in V$ s.t. $(\vkey, \vkey'|\cdot)\in E$, it returns \nul.
	\item An $\getsp(\vkey)$, if $\vkey \in V$, returns a set $S(\vkey) = \{d(\vkey_i)\}_{\vkey_i\in V}$, where $d(\vkey_i)$ is the summation of the weights of the edges\footnote{We limit our discussion to positive edge-weights only.} on the shortest-path between $\vkey$ and $\vkey_i$ if $\vkey_i \hookleftarrow \vkey$, and $d(\vkey_i) = \infty$, if $\vkey_i \not\hookleftarrow \vkey$. Note that $d(\vkey) = 0$. There can be multiple paths between $\vkey$ and $\vkey_i$ with the same sum of edge-weights. If $\vkey \notin V$, it returns \nul.
	\item A $\getbc(\vkey)$ returns the betweenness centrality of \vkey as defined before, if $\vkey \in V$. It returns \nul if $\vkey \notin V$.
\end{enumerate}

\input{figs/dsstruct.tex}

\subsection*{Design Requirements}
Firstly, because we intend to implement efficient dynamic modifications in an unbounded graph, we choose its adjacency list representation. An adjacency list of a graph, essentially, translates to a composition of dictionaries: one \textit{\vlist} and as many \textit{\elists} as the number of vertices. The members of the \vlist correspond to $v\in V$, wherein each of them has an associated \elist corresponding to $E_v$. Implementing the ADT requires the \vlist and \elists offer membership queries along with addition/removal of keys. Significantly, the addition/removal need to be \nbk.


Next, in essence, the operations \getbfs, \getsp, and \getbc, also termed as 
\textit{queries}, are special partial snapshots of the data 
structure. Naturally, they scan (almost) the entire graph. This requires 
supporting the scan while keeping track of the 
visited nodes. This is more important for traversals to happen
non-recursively. 

\subsection*{Data Structure Components}
Keeping the above requirements in view, we build the data structure based on a composition of  a \lf hash-table implementing the \vlist, and \lf BSTs implementing the \elists. On a skeleton of this composition, we include the design components for efficient traversals and (partial) snapshots. This is a more efficient design as compared to Chatterjee et al.'s approach \cite{Chatterjee+:NbGraph:ICDCN-19} where the component dictionaries are implemented using \lf linked-lists only.

More specifically, the nodes of the \vlist are instances of the class \vnode, see \figref{struct-evnode-graph}(a). A \vnode contains the key of the corresponding vertex along with a pointer to a BST 
implementing its \elist. The most important member of a \vnode is a pointer to an instance of the class \opstruct, which serves the purpose of anchoring the traversals as described above. 

The \opstruct class, see \figref{struct-evnode-graph}(b), encapsulates an array \visitedarray of the size equal to the number of threads in the system, a counter \ecount and other algorithm specific indicators, which we describe in \secref{applications} while specifying the queries. An element of \visitedarray simply keeps a count of the number of times the node is visited by a query performed by the corresponding thread. The counter \ecount is incremented every time an outgoing edge is added or removed at the vertex. This serves an important purpose of notifying a thread if the same edge is removed and 	added since the last visit. 

The class \enode, see \figref{struct-evnode-graph}(b), structures the nodes of an \elist. It encapsulates a key, the edge-weight, the left- and right-child pointers and a pointer to the associated \vnode where the edge terminates; the key in the \enode is that of the \vnode; thus each \enode delegates a directed edge.

The \vnodes are bagged in a linked-list being referred to by a pointer from the \buckets, see \figref{struct-evnode-graph}(a). A resizable hash-table is constructed of the arrays of these \buckets, wherein arrays are linked in the form of a linked-list of \hnodes. 
 
At the bare-bone level, our dynamically resizable \lf  hash-table derives from Liu et al. \cite{Liu+:LFHash:PODC:2014}, whereas the BSTs, implementing the \elists, are based on \lf internal BST of 
Howley et al. \cite{Howley+:NbkBST:SPAA:2012}. We introduce the \opstruct fields in hash-table nodes. To facilitate non-recursive traversal in the \lf BST, we use stacks. As we will explain later, aligning the operations of the hash-table to the state of \opstruct therein brings in nontrivial challenges.

The last but a significant component of our design is the class \treenode, see \figref{struct-evnode-graph}(b). It encapsulates the information which we use to validate a scan of the graph to output a consistent specialized partial snapshot. More specifically, it packs the pointers to \vnodes visited during a scan along with two pointers \bnext and \texttt{p} to keep track of the order of their visit. The field \lecount records the \ecount counter of the corresponding visited \vnode, which enables checking if the visited \vnode has had any addition or removal of an edge since the last visit.\\

\subsection*{\Nbk Data Structure Construction}
Having these components in place, we construct a non-blocking graph data structure in a modular fashion. Refer to \figref{struct-evnode-graph}(d) depicting a partial implementation of a small directed graph shown in \figref{struct-evnode-graph}(c). The \enodes, shown as circles in \figref{struct-evnode-graph}(d), with their children and parent pointers make \lf internal BSTs corresponding to the \elists. For simplicity we have only shown the outgoing edges of vertex 5 in \figref{struct-evnode-graph}(d) while the edges of other vertices are represented by small triangles. Thus, whenever a vertex has outgoing edges, the corresponding \vnode, shown as small rectangles therein, has a non-null pointer pointing to the root of a BST of \enodes. The \vnodes themselves make sorted \lf linked-lists connected to the buckets of a hash-table. The buckets are cells of a bucket-array implementing the \lf hash-table. When required, we add/remove bucket-arrays for an unbounded resizable dynamic design. The \lf~ \vnode-lists have two sentinel \vnodes: \vh and \vt initialized with keys $\text{-}\infty$ and $\infty$, respectively. 

We adopt the well-known technique of \textit{pointer marking} -- using a single-word \CAS -- via bit-stealing 
\cite{Howley+:NbkBST:SPAA:2012,Liu+:LFHash:PODC:2014} to perform lazy 
non-blocking removal of nodes. More pointedly, on a common x86\text{-}64 
architecture, memory has a 64-bit boundary and the last three least significant 
bits are unused; this allows us to use the last one significant bit of a 
pointer to indicate first a \textit{logical removal} of a node and thereafter 
cleaning it from the data structure. Specifically, an \hnode, a \vnode, and an \enode is logically removed by marking its \hpred, \vnext, and \eleft pointer, respectively. We call a node \textit{alive} which is not logically removed.

\subsection*{Data Structure Invariants}
To prove the correctness of the presented algorithm, we fix the 
\textit{invariants} corresponding to a consistent state of the composite data structure: 
\begin{enumerate}[leftmargin=5.5mm,label=\emph{\alph*})]
\item each \elist maintains a BST order based on the \enode's key $\ekey$, and alive \enodes are reachable from \enext of the corresponding \vnode, 
\item a \vnode that holds a pointer to a BST containing any alive \enode is itself alive, 
\item each alive \vnode is reachable from \vh and \vlist{s} connected to buckets are sorted based on the \vnode's keys $\vkey$, and 
\item an \hnode which contains a \bucket holding a pointer to an alive \vnode is itself alive and an alive \hnode is always connected to the linked-list of \hnodes.
\end{enumerate}

\subsection*{Correctness and Progress Guarantee}
To prove \textit{\lbty} \cite{Herlihy+:lbty:TPLS:1990}, we describe the execution generated by the \ds as a collection of 
\mth invocation and response events. We assign an atomic step between the invocation and response as the \emph{linearization point} (\emph{\lp}) of a \mth call (operation). Ordering the operations by their \lp{s} provide a sequential history of the execution. We prove the correctness of the data structure by assigning a sequential history to an arbitrary execution which is valid i.e. it maintains the invariants. Furthermore, we argue that the data structure is \nbk by showing that the queries would return in finite number of steps if no update operation happens and in an arbitrary execution at least one update operation returns in finite number of steps.

The progress properties specify when a thread invoking operations on the shared memory objects completes in the presence of other concurrent threads. In this context, we provide the graph implementation with \mth{s} that satisfy wait-freedom, based on the definitions in Herlihy and Shavit \cite{Herlihy+:OnNatProg:opodis:2001}. A \mth of a concurrent data-structure is \wf if it completes in finite number of steps. A \ds implementation is \wf if all its \mth{s} are \wf. This ensures per-thread progress and is the most reliable non-blocking progress guarantee in a concurrent system. A \ds is \lf if its \mth{s} get invoked by multiple concurrent threads, then one of them will complete in finite number of steps. 

%% file: figs/dsstruct.tex
\begin{figure*}[!ht]
    \centering
   	\captionsetup{font=footnotesize}
   	\begin{subfigure}{.28\textwidth}
	\captionsetup{font=footnotesize, justification=centering}
	\begin{footnotesize}
			\begin{tabbing}
			\hspace{0.05in} \= \hspace{0.05in} \= \hspace{0.05in} \=  \hspace{0.05in} \= \\
			\\
		  \> {\bf struct \bucket \{} \\
			\> \> {\vnode $vn$;} \\
		    \> \}\\  
		    
		    \> {\bf struct \hnode \{} \\
			\> \> {\bucket [];}\\
		    \> \> {int \Hsize;} \\
		    \> \> {\hnode} {\hpred;}\\
		    \> \}\\  
		    	\> {\bf class \vnode \{} \\
				\> \> \texttt{int}  $\vkey$; 
				\\
				\> \> {\vnode~ \vnext;}
				\\
				\> \> {\enode~ \enext;} 
				\\
				\> \> {\opstruct~ \opitem;}
				\\
			\> \} \\
		\end{tabbing}
	\vspace{-0.3in}
	\end{footnotesize}
	\label{fig:struct-vnode}
	\caption{}
	\end{subfigure}
	\begin{subfigure}{.24\textwidth}
	\captionsetup{font=footnotesize, justification=centering}
	\begin{footnotesize}
					\begin{tabbing}
						\hspace{0.05in} \= \hspace{0.05in} \= \hspace{0.05in} \=  \hspace{0.05in} \= \\
			\\
			\> {\bf class \opstruct \{} \\
			 \> \> {int~ \ecount, \visitedarray [];} 
		    \\
			\> \> $\cdots$ // Other fields\\
		    \> \} \\
		\> {\bf class \enode \{} \\
			\> \> \texttt{int}  $\ekey$;   
			\\
			\> \> \texttt{double}  $\eweight$;   
			\\
			\> \> {\vnode~ \pointv;} 
			\\
			\> \> {\enode~ \eleft, \eright;} 
			\\
		\> \}\\
			\> {\bf class \treenode \{} \\
							\> \> {\tt \vnode~   $n$;} 
							\\
							\> \> {\tt \treenode~ \bnext, p;}\\ 
							\> \> {int    \ecount;} 
			                \\
						\> \} \\
		\end{tabbing}
	 \vspace{-0.5in}
	\end{footnotesize}
	\label{fig:struct-opitem}
		\caption{}
		\end{subfigure}
		\begin{subfigure}{.45\textwidth}
		\centerline{\scalebox{0.55}{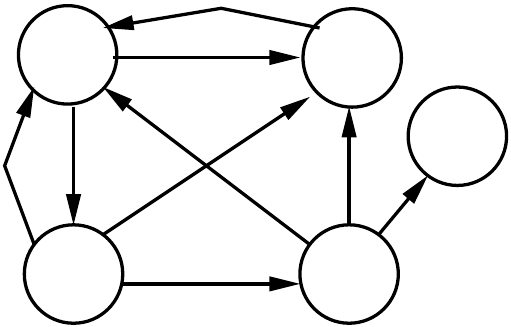}}
		 \vspace{-0.1in}
		\label{subfig:graph}
		\caption{}
		\centerline{\scalebox{0.6}{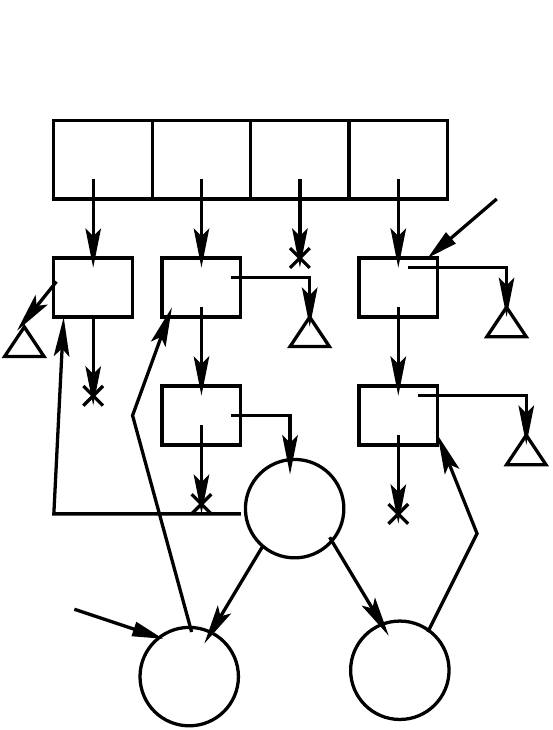}}
		\label{subfig:graph-rep}
		\caption{}
	\end{subfigure}
	\vspace{-3mm}
	\caption{(a) and (b) Data structure components, (c) A sample directed graph, (d) Our implementation of (c) as a composition of a \lf hash-table and \lf BSTs.}
	\label{fig:struct-evnode-graph}
\end{figure*}

%% file: figs/digraph.pdf_t
\begin{picture}(0,0)%
\includegraphics{figs/digraph.pdf}%
\end{picture}%
\setlength{\unitlength}{4144sp}%
\begingroup\makeatletter\ifx\SetFigFont\undefined%
\gdef\SetFigFont#1#2#3#4#5{%
  \reset@font\fontsize{#1}{#2pt}%
  \fontfamily{#3}\fontseries{#4}\fontshape{#5}%
  \selectfont}%
\fi\endgroup%
\begin{picture}(2332,1482)(2634,-2551)
\put(4636,-1771){\makebox(0,0)[lb]{\smash{{\SetFigFont{12}{14.4}{\rmdefault}{\bfdefault}{\updefault}{\color[rgb]{0,0,0}$4$}%
}}}}
\put(2881,-2401){\makebox(0,0)[lb]{\smash{{\SetFigFont{12}{14.4}{\rmdefault}{\bfdefault}{\updefault}{\color[rgb]{0,0,0}$3$}%
}}}}
\put(2881,-1366){\makebox(0,0)[lb]{\smash{{\SetFigFont{12}{14.4}{\rmdefault}{\bfdefault}{\updefault}{\color[rgb]{0,0,0}$1$}%
}}}}
\put(4141,-2401){\makebox(0,0)[lb]{\smash{{\SetFigFont{12}{14.4}{\rmdefault}{\bfdefault}{\updefault}{\color[rgb]{0,0,0}$5$}%
}}}}
\put(4186,-1411){\makebox(0,0)[lb]{\smash{{\SetFigFont{12}{14.4}{\rmdefault}{\bfdefault}{\updefault}{\color[rgb]{0,0,0}$7$}%
}}}}
\end{picture}%

%% file: figs/digraph-rep.pdf_t
\begin{picture}(0,0)%
\includegraphics{figs/digraph-rep.pdf}%
\end{picture}%
\setlength{\unitlength}{4144sp}%
\begingroup\makeatletter\ifx\SetFigFont\undefined%
\gdef\SetFigFont#1#2#3#4#5{%
  \reset@font\fontsize{#1}{#2pt}%
  \fontfamily{#3}\fontseries{#4}\fontshape{#5}%
  \selectfont}%
\fi\endgroup%
\begin{picture}(2519,3324)(4839,-3790)
\put(5176,-961){\makebox(0,0)[lb]{\smash{{\SetFigFont{11}{13.2}{\rmdefault}{\bfdefault}{\updefault}{\color[rgb]{0,0,0}$B[0]$}%
}}}}
\put(5671,-961){\makebox(0,0)[lb]{\smash{{\SetFigFont{11}{13.2}{\rmdefault}{\bfdefault}{\updefault}{\color[rgb]{0,0,0}$B[1]$}%
}}}}
\put(6076,-961){\makebox(0,0)[lb]{\smash{{\SetFigFont{11}{13.2}{\rmdefault}{\bfdefault}{\updefault}{\color[rgb]{0,0,0}$B[2]$}%
}}}}
\put(6481,-961){\makebox(0,0)[lb]{\smash{{\SetFigFont{11}{13.2}{\rmdefault}{\bfdefault}{\updefault}{\color[rgb]{0,0,0}$B[3]$}%
}}}}
\put(5176,-1861){\makebox(0,0)[lb]{\smash{{\SetFigFont{12}{14.4}{\rmdefault}{\bfdefault}{\updefault}{\color[rgb]{0,0,0}$4$}%
}}}}
\put(5671,-1861){\makebox(0,0)[lb]{\smash{{\SetFigFont{12}{14.4}{\rmdefault}{\bfdefault}{\updefault}{\color[rgb]{0,0,0}$1$}%
}}}}
\put(5671,-2401){\makebox(0,0)[lb]{\smash{{\SetFigFont{12}{14.4}{\rmdefault}{\bfdefault}{\updefault}{\color[rgb]{0,0,0}$5$}%
}}}}
\put(6616,-3571){\makebox(0,0)[lb]{\smash{{\SetFigFont{12}{14.4}{\rmdefault}{\bfdefault}{\updefault}{\color[rgb]{0,0,0}$7$}%
}}}}
\put(6121,-2851){\makebox(0,0)[lb]{\smash{{\SetFigFont{12}{14.4}{\rmdefault}{\bfdefault}{\updefault}{\color[rgb]{0,0,0}$4$}%
}}}}
\put(5671,-3616){\makebox(0,0)[lb]{\smash{{\SetFigFont{12}{14.4}{\rmdefault}{\bfdefault}{\updefault}{\color[rgb]{0,0,0}$1$}%
}}}}
\put(4951,-3211){\makebox(0,0)[lb]{\smash{{\SetFigFont{12}{14.4}{\rmdefault}{\bfdefault}{\updefault}{\color[rgb]{0,0,0}$\enode$}%
}}}}
\put(6976,-1366){\makebox(0,0)[lb]{\smash{{\SetFigFont{12}{14.4}{\rmdefault}{\bfdefault}{\updefault}{\color[rgb]{0,0,0}$\vnode$}%
}}}}
\put(5176,-601){\makebox(0,0)[lb]{\smash{{\SetFigFont{10}{12.0}{\rmdefault}{\bfdefault}{\updefault}{\color[rgb]{0,0,0}Hash function: $f(x)$ = $x$ \Mod $4$}%
}}}}
\put(6571,-2401){\makebox(0,0)[lb]{\smash{{\SetFigFont{12}{14.4}{\rmdefault}{\bfdefault}{\updefault}{\color[rgb]{0,0,0}$7$}%
}}}}
\put(6526,-1816){\makebox(0,0)[lb]{\smash{{\SetFigFont{12}{14.4}{\rmdefault}{\bfdefault}{\updefault}{\color[rgb]{0,0,0}$3$}%
}}}}
\end{picture}%

%% file: graph-algo.tex
\input{code/graphfrmwk-hash-code.tex}

In this section, we describe a \nbk algorithm that implements the ADT $\mathcal{A}$. 
The operations $\mathcal{M}:=\{\addv, \remv, \conv, \adde, \reme, \cone\}\subset\mathcal{A}$ use the interface of the hash-table and BST with interesting non-trivial adaptation to our purpose. In the permitted space we describe the execution, correctness and progress property of the operations  $\mathcal{Q}:=\{\getbfs, \getsp, \getbc\}\subset\mathcal{A}$. To de-clutter the presentation, we encapsulate the three queries in a unified framework. The framework comes with an interface operation \getgphalgo. \getgphalgo is specialized to the requirements of the three queries. The functionality of \getgphalgo and its specializations are presented in pseudo-code in Figures \ref{fig:graph-framework}, \ref{fig:sssp-methods}, and \ref{fig:bc-methods}.

\noindent \textbf{Pseudo-code convention:}
We use $p.x$ to denote the member field $x$ of a class object pointer $p$. To 
indicate multiple return objects from an operation we use $\langle x_1, 
x_2,\ldots,x_n\rangle$. To represent pointer-marking, we 
define three procedures: (a) $\isMarked(p)$ returns \tru if the last 
significant bit of the pointer $p$ is set to 
$1$, else, it returns \fal, (b) $\MarkedRef(p)$ sets last 
significant bit of $p$ to $1$, and (c) $\unMarkedRef(p)$ sets the same to $0$. 
An invocation of $\createv(\vkey)$, $\createe(\ekey)$ and 
$\createtreenode(\vkey)$, creates a new \vnode with key $\vkey$, a new \enode 
with key $\ekey$ and a new \treenode with a \vnode $v(\vkey)$ respectively. For 
a newly created \vnode, \enode and 
\treenode the pointer fields are initialized with \nul value.

The execution pipeline of \getgphalgo is presented at lines \ref{getgphalgostart} to \ref{getgphalgoend} in Fig. \ref{fig:graph-framework}. \getgphalgo intakes a query vertex $v$. It starts with checking if $v$ is alive at Line \ref{checkvisitedstart}. In the case $v$ was \textit{not alive}, it returns \nul. For this execution case, which results in \getgphalgo returning \nul, the \lp is at the atomic step (a) where \getgphalgo is invoked in case $v$ was not in the data structre at that point, and (b) where $v$ was logically removed using a \CAS in case it was alive at the invocation of \getgphalgo. 

Now, if the query vertex $v$ is alive, it proceeds to perform the method \scan, Line \ref{scanstart} to \ref{scanend}. \scan, essentially, repeatitively, performs (specialized partial) snapshot collection of the data structure along with comparing every consecutive pair of scans, stopping when a consecutive pair of collected snapshots are found identical. Snapshot collection is structured in the method \treeclt, Line \ref{treecltstart} to \ref{treecltend}, whereas, comparison of collected snapshots in performed by the method \comparetree, Line \ref{comparetreestart} to \ref{comparetreeend}. 

Method \treeclt performs a BFS traversal in the data structure to collect pointers to the traversed \vnodes, thereby forming a tree. A cell of \visitedarray corresponding to thread-id is marked on visiting it; notice that it is adaptation of the well-known use of node-dirty-bit for BFS \cite{Cormen+:IntroAlgo:2009}. The traversal over \vnodes in facilitated by a queue: Line \ref{treecltqdef}, whereas, exploring the outgoing edges at each \vnode, equivalently, traversing over the BST corresponding to its \elist uses a stack: Line \ref{treecltstdef}. The snapshot collection for the queries \getbfs and \getbc are identical. For \getsp, where edge-weights are to be considered, the snapshot collection is optimized in each consecutive scan based on the last collection; we detail them in the next section. At the core, the collected snapshot is a list of \treenode{s}, where each \treenode contains a pointer to a \vnode, pointers to the next and previous \treenode{s} and the value of the \ecount field of the \opstruct of the \vnode.

Method \comparetree essentially compares two snapshots in three aspects:  whether the collected \treenode{s} contain (a) pointers to the same \vnodes (b) have the same \treenode being pointed by previous and next, and (c) have the same \ecount. The three checks ensure that a consistent snapshpt is the one which has its collection lifetime \textit{not concurrent} to (a) a vertex either added to or removed from it, (b) a path change by way of addition or removal of an edge, and, (c) an edge removed and then added back to the same position, respectively. Thus, at the completion of these checks, if two consecutive snapshots match, it is guranteed to be unchanged during the time of the last two \treeclt operations. Clearly, we can put a linearization point just after the atomic step where the last check is done: Line \ref{lp1} or \ref{lp2}, where it returns \comparetree.

Now, it is clear that any $Q\in\mathcal{Q}$ does not engage in helping any other operation. Furthermore, an $M\in\mathcal{M}$ does not help a $Q\in\mathcal{Q}$. Thus, given an execution $E$ as a collection of arbitrary $O\in\mathcal{A}$, by the fact that the data-structures hash-table and BST are lock-free, and whenever no \putv, \pute, \remv, and \reme happen, a $Q\in\mathcal{Q}$ returns, we infer that the presented algorithm is \nbk.

%% file: code/graphfrmwk-hash-code.tex
\begin{figure*}[!htb]
\captionsetup{font=scriptsize}
\begin{multicols}{2}
\begin{algorithmic}[1]
		\scriptsize
        \renewcommand{\algorithmicprocedure}{\textbf{Operation}}
		\Procedure{$\getgphalgo$($v$)}{}\label{getgphalgostart}
        \State{$tid$ $\gets$ $\gettid$();}\linecomment{//Thread-id to assign a cell of \visitedarray array.}
	\If{($\isMarked(v)$)} 
	\label{lin:getgphalgo-checkmarked} 
	     \State {return  \nul; } \linecomment{//Vertex is not present.}
        \EndIf
        \State{ $\langle st, stree \rangle$ $\gets$ $\scan(v, tid)$;} \label{lin:rcbl-scan}\linecomment{//Invoke Scan }
         \If{($st$ = \tru)} 
          \State return {$stree$;}
         \Else 
         \State{return \nul;}
         \EndIf 
		\EndProcedure\label{getgphalgoend}
        \algstore{getgphalgo}
\end{algorithmic}	
    \hrule
	\begin{algorithmic}[1]
	\algrestore{getgphalgo}
		\scriptsize
		\renewcommand{\algorithmicprocedure}{\textbf{Method}}
		
		\Procedure{$\scan$($v$, $tid$)}{}\label{scanstart}
       \State{list$<\treenode>$ot, nt ; } \linecomment{//Trees to hold the nodes}
          \State{$ot$ $\gets$ \treeclt($v$, $tid$); }\linecomment{//$1^{st}$ Collect}
        \While{(\tru)} \label{whilescan} \linecomment{//Repeat the tree collection}
           \State{$nt$ $\gets$ \treeclt($v$, $tid$); } \linecomment{//$2^{nd}$ Collect}
           \If{(\comparetree($ot, nt$))} 
           \label{lin:scan-comparetree}
           \State{return $\langle \tru, nt \rangle$;}\linecomment{//Two collects are equal}
           \EndIf
            \State{ot $\gets$ nt;} \linecomment{//Retry if two collects are not equal}
        \EndWhile
		\EndProcedure\label{scanend}
        \algstore{scan}
\end{algorithmic}
\hrule
\begin{algorithmic}[1]
	\algrestore{scan}
		\scriptsize
		\renewcommand{\algorithmicprocedure}{\textbf{Method}}
		\Procedure{$\comparetree$($ot, nt$)}{}\label{comparetreestart}
	 \If{($ot$ = \nul$\vee$ $nt$ = \nul)} 
	     \State {return $\fal$ ; }
        \EndIf
        \State{ $oit$ $\gets$ $ot$.head, $nit$ $\gets$ $nt$.head;} 
        \While{($oit$ $\neq$ $ot$.tail $\wedge$ $nit$ $\neq$ $nt$.tail )}
        \If{($oit$.n $\neq$ $nit$.n $\vee$ $oit.\ecount$ $\neq$ $nit.\lecount$ $\vee$ $oit$.p $\neq$ $nit$.p)}\label{lp1} 
          \State{return $\fal$ ; } 
        \EndIf
        \State{$oit$ $\gets$  $oit$.\bnext; $nit$ $\gets$  $nit$.\bnext;}
        \EndWhile
         \If{($oit$.n $\neq$ $nit$.n $\vee$ $oit$.\ecount $\neq$ $nit$.\lecount $\vee$ $oit$.p $\neq$ $nit$.p)}\label{lp2} 
         \State {return $\fal$ ; }
         \Else {\hspace{2mm}return $\tru$ ;  } 
         \EndIf
		\EndProcedure\label{comparetreeend}
        \algstore{comparetree}
\end{algorithmic}

\hrule
\begin{algorithmic}[1]
	\algrestore{comparetree}
		\scriptsize
		\renewcommand{\algorithmicprocedure}{\textbf{Method}}
		\Procedure{ $\checkvisited$($adjn, tid, count$)}{}\label{checkvisitedstart}
	 \If{($adjn$.\opitem.\visitedarray[$tid$] = $count$)}
	     \State {return $\tru$;} \linecomment{//\vnode is not visited}
	     \Else {\hspace{2mm}return $\fal$ ; }
        \EndIf
		\EndProcedure\label{checkvisitedeend}
        \algstore{checkvisited}
\end{algorithmic}
\hrule	
	\begin{algorithmic}[1]
	\algrestore{checkvisited}
\renewcommand{\algorithmicprocedure}{\textbf{Operation}}	
\scriptsize
	\renewcommand{\algorithmicprocedure}{\textbf{Method}}
	\Procedure{$\treeclt$($v$, $tid$)}{}\label{treecltstart}
         \State{queue $<$\treenode$>$ $que$;}\label{treecltqdef}  
         \State{list$<$\treenode$>$$st$;} {\tcount$\gets$\tcount$+$ 1;} 
         \State{$v$.\opitem.\visitedarray[$tid$] $\gets$ \tcount ;}\label{lin:treelt-visited} \linecomment{//Mark it visited} 
        \State{$sn$$\gets$new $\createtreenode$($v$,\nul,\nul,$v.oi$.\ecount);} 
        \State{$st$.$\add$($sn$);}{$que$.enque($sn$);}\label{lin:bfsTree:insert1}  
        \While{($\neg  que$.empty())} 
        \State{ $cvn$ $\gets$ $que$.deque();}  
         \If{(\isMarked($cvn$))} {\cntu;}
         \EndIf
        \State{$itn$ $\gets$ $cvn$.n.\enext;} 
        \State{stack $<$\enode $>$ $S$;}\label{treecltstdef} 
        \While{($itn$ $\vee$ $\neg S$.empty())}
        \While{($itn$ )}
         \If{($\neg$$\isMarked$($itn$))}
         \State{$S$.push($itn$)}; \linecomment{// push the \enode}
         \EndIf
         \State{$itn$ $\gets$ $itn.\eleft$;} 
        \EndWhile
        \State{$itn$ $\gets$ $S$.pop()};
        \If{($\neg$$\isMarked$($itn$))} 
        \State{$adjn$ $\gets$ $itn$.\pointv;}\linecomment{//Get the \vnode}
        \If{($\neg$\isMarked($adjn$))} \label{lin:treeclt-marked-adjn} \linecomment{//Validate it}
         \If{($\neg$\checkvisited($adjn$, $tid$, \tcount))}
       \State{$adjn$.\opitem.\visitedarray[$tid$] $\gets$ \tcount;} \linecomment{//Mark it visited} 
       \State{$sn$ $\gets$ new $\createtreenode$($adjn$,$cvn$,\nul,$adjn.oi$.\ecount);} 
       \State{$st$.$\add$($sn$);}\label{lin:bfsTree:insert2} 
        \State{$que$.enque($sn$);}
        
        \EndIf
        \EndIf
        \EndIf
        \State{$itn$ $\gets$ $itn.\eright$;} 
        \EndWhile
    \EndWhile
     \State{return $st$;}\label{lin:return-bfstree} 
		\EndProcedure\label{treecltend}
		\algstore{treeclt} 
\end{algorithmic}
\end{multicols}
\vspace{-3mm}
\caption{\normalsize Framework interface operation for graph queries.} \label{fig:graph-framework}
\end{figure*}

%% file: applications.tex
Having described the execution pipeline of \getgphalgo, here we specialize it to the queries \bfs, \sssp, and \bc. As mentioned before, the \opstruct class is inducted with query specific fields for $O \in \mathcal{Q}:=\{\getbfs, \getsp, \getbc\}$, which sits at the core of its implementation.


\subsection{\BFS} \label{subsec:bfs}
\input{Graph-Applications/bfs}

\input{code/bfs-dia-hash-code.tex}


\subsection{Single Source Shortest Paths} \label{subsec:sssp-bellmanford}
\input{Graph-Applications/sssp-bellmanford}
\subsection{Betweenness Centrality} \label{subsec:bc}
\input{code/bc-hash-code.tex}

\input{Graph-Applications/bc}

%% file: Graph-Applications/bfs.tex
\ignore{
For a $\getbfs\in\mathcal{Q}$, the \opstruct class has only two fields: \ecount and \visitedarray, whose functionality we have already described earlier. Given an unweighted graph $G=(V,E)$, $\getbfs(\vkey)$ starts from the \textit{source} vertex \vkey to traverse the vertices $u \in V$, which are reachable from $\vkey$, in BFS order (see section \ref{sec:intro}) to them. The method \treeclt as described in Figure \ref{fig:graph-framework} exactly performs a BFS traversal. The return of a \treeclt is a \textit{BFS-tree} rooted at \vkey. To ensure \lbty, we perform repeated \treeclt, and on a match of two consecutively collected BFS-trees, $\getbfs(\vkey)$ terminates returning the output of the last \treeclt. Note that, in the case $\vkey$ is a leaf-node, the output is $\vkey$ itself (essentially, $st$ is always \tru for a \getbfs at line \ref{lin:rcbl-scan}).
}

Given an unweighted graph $G=(V,E)$, $\getbfs(\vkey)$ starts from the \textit{source} vertex \vkey to traverse the vertices $u \in V$, which are reachable from $\vkey$, in BFS order (see section \ref{sec:intro}) to them. 

\ignore{In the worst case it requires $O(|V|+|E|)$ operations (see \cite{Cormen+:IntroAlgo:2009}). BFS identifies the distance using the least number of edges from $\vkey$ to the vertices which are reachable from $\vkey$. At the end, it returns a BFS tree with root $\vkey$ that holds all reachable vertices. There have been several works has been done both in theory and practice on developing fast parallel \bfs algorithms both for distributed and shared-memory architecture \cite{Leiserson:WPBFS:SPAA:2010,Buluc+:PBFS:SC:2011, McLaughlin+:FBFS:ICPADS:2015, You+:BFS:ICPP:2014, Munguia+:TBFS:HiPC:2012, Bader+:BFS:ICPP:2006}.
}
Implementation of $\getbfs\in\mathcal{Q}$ is shown in \figref{bfs-methods}. The \opstruct class has only two fields: \ecount and \visitedarray, whose functionality we have already described earlier. A $\getbfs(\vkey)$ operation, at Lines \ref{getbfsstart} to \ref{getbfsend}, begins by validating the presence of the $\vkey$ in the hash table and is unmarked. If the validation fails, it returns \nul. Once the validation succeeds, \bfsscan method is invoked. A \bfsscan method, at Lines \ref{bfsscanstart} to \ref{bfsscanend}, performs repeated \bfstclt(BFS Tree Collect). A \bfstclt method, at Lines \ref{bfstcltstart} to \ref{bfstcltend}, performs non-recursive \bfs traversal starting from a source vertex $\vkey$. In the process of traversal, it explores all unmarked reachable \vnodes through corresponding unmarked \enodes. However, it keeps adding all these \vnodes in the tree $bt$ (see Line \ref{lin:bfstclt-add1} and \ref{lin:bfstclt-add2}). The traversal of reachable \vnodes are processed through a queue (Line \ref{lin:bfstree-que}), whereas, all outgoing neighbourhood edges of each \vnode is processed with a stack (Line \ref{lin:bfstree-stack}). After exploring all reachable \vnodes, the \bfstclt terminates by returning $bt$ to the \bfsscan. 

After each two consecutive \bfstclt, we invoke \comparetree to compare the two trees are equal or not. If both trees are equal, then \bfsscan terminates returning the output of the last \bfstclt along with a boolean status value \tru to the $\getbfs$ operation. If two consecutive \bfstclt do not match in the \comparetree, then we discard the older tree and restart the \bfstclt. Note that, in the case $\vkey$ is a leaf-node, the output is $\vkey$ itself (essentially, $st$ is always \tru for a \getbfs at Line \ref{lin:bfs-rcbl-scan}).

The \comparetree method shown in \figref{graph-framework}, from Line \ref{comparetreestart} to \ref{comparetreeend}. It compares two trees by comparing the individual BFS nodes's reference ($n$), parent ($p$) and  \lecount values. It starts from the head \bfsnode and follows with the next pointer, \bnext, until it reaches the last \bfsnode or any mismatch occurrs at any \bfsnode. It returns \fal if any mismatch occurrs. Otherwise returns \tru implying that all the nodes of the trees are same. 

%% file: code/bfs-dia-hash-code.tex
\begin{figure*}[!htb]
\begin{multicols}{2}
    \begin{algorithmic}[1]
	\algrestore{treeclt}
	\scriptsize
        \renewcommand{\algorithmicprocedure}{\textbf{Operation}}	
	\Procedure{$\getbfs$($\vkey$)}{}\label{getbfsstart}
    \State{$tid$ $\gets$ $\gettid$();}\linecomment{//Get the thread id}
	  \If{($\isMarked$($\vkey$)}
	  \linecomment{//Validate the vertex}
	 \State{return \nul;} \linecomment{//$\vkey$ is not present.}
        \EndIf
         \State{list $<\bfsnode>$  $bt$; } \linecomment{//List to hold the \bfstree}
         \State{$\langle st,bt\rangle$$\gets$$\bfsscan$($\vkey$,$tid$);}\label{lin:bfs-rcbl-scan}\linecomment{//Invoke the scan}
         \If{($st$ = \tru)} 
        \State{return $bt$;} \linecomment{// Return the \bfstree}
         \EndIf 
	\EndProcedure\label{getbfsend}
        \algstore{getbfs}
\end{algorithmic}	
    \hrule
    \scriptsize
	\begin{algorithmic}[1]
	\algrestore{getbfs}
	\scriptsize
	\renewcommand{\algorithmicprocedure}{\textbf{Method}}
	\Procedure{$\bfsscan$($\vkey$, $tid$)}{}\label{bfsscanstart}
        \State{list $<\bfsnode>$ $ot$, $nt$ ; } 
          \State{$ot$ $\gets$ $\bfstclt$($\vkey$, $tid$); }\linecomment{//$1^{st}$ Collect}
        \While{(\tru)} \label{bfswhilescan1} 
           \State{$nt$ $\gets$ \bfstclt($\vkey$, $tid$); } \linecomment{//$2^{nd}$ Collect}
           \If{($\comparetree$($ot$, $nt$))}  
           \State{return $\langle$\tru,$nt$$\rangle$;}
           \EndIf
           \State{$ot$$\gets$$nt$;}\label{bfsrescan1}
    \EndWhile
		\EndProcedure\label{bfsscanend}
        \algstore{bfsscan}
\end{algorithmic}
\hrule
	\begin{algorithmic}[1]
	\algrestore{bfsscan}
        \renewcommand{\algorithmicprocedure}{\textbf{Method}}	
	\scriptsize
    	\Procedure{$\bfstclt$($\vkey$, $tid$)}{}\label{bfstcltstart}
        \State{\bfsnode $bt$;} 
        \linecomment{//List to keep of the visited nodes}
	   \State{queue $<$\bfsnode$>$ $que$;} \linecomment{//Queue used for traversal} \label{lin:bfstree-que}
         \State{\tcount $\gets$ \tcount $+$ 1; } 
         \State{v.\opitem.\visitedarray[$tid$] $\gets$ \tcount ;} \linecomment{//Mark it visited} 
         
         \State{$bn$ $\gets$ new $\bfsnode$($\vkey$, \nul, \nul,$\vkey$.\opitem.\ecount);}
        \State{$bt$.$\add$($bn$);} 
        \label{lin:bfstclt-add1}
         \State{$que$.enque($bn$);} \linecomment{//Push the initial node into the $que$}
        \While{($\neg$ $que$.empty())} 
        \State{$cvn$ $\gets$ $que$.deque();} \linecomment{// Get the front node}
        \If{($\neg \isMarked$($cvn$.n))} 
        \State{\cntu;} \linecomment{// If marked then continue} 
        \EndIf
        \State{$itn$ $\gets$ $cvn$.n.\enext;} \linecomment{//Get the root \enode from its \elist}
        \State{stack $<$\enode $>$ $S$;} \label{lin:bfstree-stack}\linecomment{// stack for inorder traversal}
        \While{($itn$ $\vee$ $\neg S$.empty())}
        \While{($itn$ )}
         \If{($\neg$$\isMarked$($itn$))}
         \State{$S$.push($itn$)}; \linecomment{// push the \enode}
         \EndIf
         \State{$itn$ $\gets$ $itn.\eleft$;} 
        \EndWhile
        \State{$itn$ $\gets$ $S$.pop()};
        \If{($\neg$$\isMarked$($itn$))}  \linecomment{//Validate it}
        \State{ $adjn$ $\gets$ $itn$.\pointv;} 
        \If{($\neg$$\isMarked$($adjn$))}\linecomment{//Validate it}
        \If{($\neg$ $\checkvisited$($tid$, $adjn$, \tcount))} 
        \State{$adjn$.\opitem.\visitedarray[$tid$] $\gets$ \tcount ;}
        \State{$bn$$\gets$ new $\bfsnode$($adjn$,$cvn$,\nul,$adjn$.\opitem.\ecount);} 
       
        \State{$bt$.$\add$($bn$);}
        \label{lin:bfstclt-add2}\linecomment{//Insert $sn$ to the $st$} 
        \State{$que$.enque($bn$);}\linecomment{//Push $sn$ into the $que$}
        \EndIf
        \EndIf
        \EndIf
         \State{$itn$ $\gets$ $itn.\eright$;}
        \EndWhile
    \EndWhile
     \State{return $bt$;} \linecomment{//The tree is returned to the \scan method}
		\EndProcedure\label{bfstcltend}
		\algstore{bfstclt}
\end{algorithmic}

\end{multicols}
	\caption{The BFS query.} \label{fig:bfs-methods}
\end{figure*}

%% file: Graph-Applications/sssp-bellmanford.tex
\ignore{The \sssp \cite{Cormen+:IntroAlgo:2009} problem defined as for a given weighted directed graph $G=(V,E, w(E))$, the weight function defined on $\mathbb{R}$, i.e., $w(E) \rightarrow$ $\mathbb{R}$, starts from a source vertex, say $\vkey$, finds the shortest path to each vertex $\vkey_i \in V$. In the end, it returns a shortest path tree with root $\vkey$ that holds the shortest path to each vertex $\vkey_i \in V$. 
For non-negative edge weights the \sssp can be solved using Dijkstra’s algorithm, in the worst case it requires $O(|E| + |V| log|V|)$ operation(see  \cite{Cormen+:IntroAlgo:2009}). Also, its parallel variants can found in \cite{Solka+:FPDijkstra:NC:1995,Jasika+:PDijkstra:MIPRO:2012}. However, for negative edge weights, the Dijkstra’s algorithm does not work. So, the Bellman-Ford algorithm \cite{Cormen+:IntroAlgo:2009} can be used, in the worst case it requires $O(|V||E|)$ operation.
}
$\getsp\in\mathcal{Q}$ builds on \getbfs and the implementation depicted in \figref{sssp-methods}. The \opstruct class here contains three fields: \ecount, \visitedarray and an array \distarray of size equal to the number of threads. For a $\getsp(\vkey)$, the distance to a visited vertex $u$ from $\vkey$ is stored at a cell of $u.\opitem.\distarray$ corresponding the traversing thread. We assume that the edge weights can be a real number: $\eweight\in\mathbb{R}$, thus, the graph $G$ can possibly contain negative edge cycles \cite{Cormen+:IntroAlgo:2009}. 
\input{code/sssp-code1.tex}

Essentially, our methodology is an adaptation of the Bellman-Ford algorithm \cite{Cormen+:IntroAlgo:2009} to our framework. Similar to \getbfs, a $\getsp(\vkey)$ operation,  Lines \ref{getspstart} to \ref{getspend}, starts by validating the presence of the $\vkey$ in the hash table and it is unmarked. After successful validation, it invokes the method \spscan, Lines \ref{spscanstart} to \ref{spscanend}, works similar to \scan described before: repeatedly collect \textit{SP-trees} using the method \sptclt, return the last one on a match of a consecutive pair of collects, see Line \ref{lin:spscan-comparetree1}. A SP-tree is a BFS-tree with some modifications described below. However, as we consider the presence of negative edge cycles, it can return \nul if a consecutive pair of \sptclt discover negative edge cycle, see Line \ref{lin:spscan-comparetree2}. The method \sptclt extends \treeclt as the following:
\begin{enumerate}[leftmargin=0.5cm,align=left,labelwidth=\parindent,labelsep=0pt,topsep=0pt,itemsep=0pt]
\item[\textsc{\textbf{Vertex distance tracking:~}}] Every time a new vertex is collected, check its distance from source and update the current recorded distance if required using method \rlx, see Lines \ref{rlxstart} to \ref{rlxend} and Line \ref{lin:sptcltrexd}. Formally, A $\rlx(tid, u, v, w(u,v)))$, at Line \ref{rlxstart} to \ref{rlxend}, method works similar as \cite{Cormen+:IntroAlgo:2009}, it checks for each edge $(u,v)$ if $v.\opitem.\distarray[tid] > u.\opitem.\distarray[tid] + w(u,v)$; if so, it sets the $v.\opitem.\distarray[tid]$ to $u.\opitem.\distarray[tid] + w(u,v)$.  

\item[\textsc{\textbf{Process BFS-tree:~}}] The method \sptclt(in Lines \ref{sptcltstart} to \ref{sptcltend}) collects \textit{SP-tree} (of \treenode{s}). First, it initializes the shortest paths \distarray(Distance Array) to all vertices to $\infty$, by invoking the method \Init(\linref{sptclt-init} except for the source vertex $\vkey$ which is initialized to $0$ (\linref{sptclt-init-vto0}). Then it traverse all unmarked reachable \vnodes through unmarked \enodes, and also a \rlx(relaxed) (\linref{sptcltrexd}) method repeatedly called. In the process of traversal it keeps adding all first time visited \vnodes in the $spt$, at \linref{sptclt-addsp-spn}), and  if it encounters already visited \vnodes which are further relaxed, then it updates the \distarray and its  parent pointer, by invoking an \updatesp method at \ref{lin:sptclt-updatespt} . At each of the \treenode{s} update to make a shortest path from source to its vertex. 

An $\updatesp$($tid$, $head$, $par$, $spn$) (Lines \ref{updatespstart} to \ref{updatespend}), starts from the $head$ node and iterate all nodes in the SP-tree until it reaches the corresponding \treenode, and then it updates the shortest path distance, at \linref{updatesp-distarray} and its parent reference pointer, at \linref{updatesp-parent}. 


\item[\textsc{\textbf{Negative cycle checking:~}}] After exploring all reachable \vnodes and Before returning the collected tree, the \sptclt invokes \checknegcycle (Check for Negative Cycle) method (at \linref{sptclt-checknegcycle}). 
A $\checknegcycle$($tid$, $spt$), at Lines \ref{checknegcyclestart} to \ref{checknegcycleend}, starts processing all \treenode{s} in the $spt$. It uses stack $S$ (\linref{checknegcycle-stack}), to process all \enodes of $spt$.$n$. For all \enodes a \rlx(relaxed) (\linref{checknegcycle-rlx}) method  called to test whether any further relaxation is possible for any \enodes, if so, it returns \fal, otherwise, returns \tru. 
\end{enumerate} 

\noindent
At the end the \sptclt method terminates by returning $spt$ and a boolean status filed(presence of neg-cycle) to the \spscan method. After each two consecutive \sptclt, \spcomparetree is invoked to compare the two trees are equal or not. If two consecutive \sptclt returns same status value \tru and both the trees are equal(see \linref{spscan-comparetree1}, then \spscan method returns \tru and $spt$ to $\getsp$ operation. However, if two consecutive \sptclt returns same status value \fal and both the trees are equal(see \linref{spscan-comparetree2}, then \spscan method returns \fal to $\getsp$ operation, which says graph has neg-cycle. If returns of two consecutive \sptclt do not match in the \spcomparetree or different status value, then we discard the older tree and restart the \sptclt.

%% file: code/sssp-code1.tex
\begin{figure*}[!htp]
\begin{multicols}{2}
\begin{algorithmic}[1]
	\algrestore{bfstclt}
		\scriptsize
\renewcommand{\algorithmicprocedure}{\textbf{Operation}}	
	\Procedure{$\getsp(\vkey)$}{}\label{getspstart} 
        \State{ $tid$ $\gets$ $\gettid$();} \linecomment{//Get the thread id}
	\If{($\isMarked$($\vkey$))}\linecomment{//Validate the vertex}
	 \State {return \nul; } \linecomment{//$\vkey$ is not present.}
        \EndIf
         \State{$\langle s, spt \rangle$ $\gets$ $\spscan$($\vkey$, $tid$);}  \label{lin:getbfssp-2c-scan} \linecomment{//Invoke the scan method}
         \If{($s$ = \tru)} 
          \State {return $spt$; } \linecomment{//Return \sptree }
         \Else 
          \State {return $\nul$; } \linecomment{// graph has neg-cycle}
         \EndIf
		\EndProcedure\label{getspend}
        \algstore{getsp}
\end{algorithmic}
\hrule
    \scriptsize
	\begin{algorithmic}[1]
	\algrestore{getsp}
		\scriptsize
		\renewcommand{\algorithmicprocedure}{\textbf{Method}}
		\Procedure{$\spscan$($\vkey$, $tid$)}{}\label{spscanstart}
         \State{list $<\spnode>$  $ospt$, $nspt$; } 
          \State{ $\langle ospt, s1 \rangle$  $\gets$ $\sptclt$($\vkey$, $tid$); }\linecomment{//$1^{st}$ Collect} 
          
        \While{(\tru)} \label{whilescangetsp} \linecomment{//Repeat the tree collection}
           \State{$\langle nspt, s2 \rangle$   $\gets$ $\sptclt$($\vkey$, $tid$);}\linecomment{//$2^{nd}$ Collect} 
           \If{($s1$ = \fal $\wedge$ $s2$ = \fal $\wedge$ $\comparetree$( $ospt$, $nspt$))} 
           \label{lin:spscan-comparetree1}
           \State{return $\langle \tru, nspt\rangle$;} 
          \Else \If{($s1$ = \tru $\wedge$ $s2$ = \tru $\wedge$ $\comparetree$($ospt$, $nspt$))}\label{lin:spscan-comparetree2}
           \State{return $\langle \fal, \nul \rangle$;}
           \EndIf
           
           \EndIf
           \State{$s1$ $\gets$ $s2$;} 
           \State{$ospt$ $\gets$ $nspt$;}\label{rescan} \linecomment{//Retry if two collects are not equal}
    \EndWhile
		\EndProcedure\label{spscanend}
        \algstore{spscan}
\end{algorithmic}
\hrule

\begin{algorithmic}[1]
	\algrestore{spscan}
	\scriptsize
	\renewcommand{\algorithmicprocedure}{\textbf{Method}}
	\Procedure{$\checknegcycle$( $tid$,
	$spt$)}{}\label{checknegcyclestart}
        \For{($itb$ $\gets$ $spt$ to $itb$.\bnext  $\neq$ \nul)} \label{lin:sptree-for1}
        \State{$cvn$ $\gets$ $itb$.n;} 
        \State{$itn$ $\gets$ $cvn$.n.\enext;} 
        \State{stack $<$\enode $>$ $S$;}\label{lin:checknegcycle-stack} \linecomment{// stack for inorder traversal}
        \While{($itn$ $\vee$ $\neg S$.empty())}
        \While{($itn$ )}
         \If{($\neg$$\isMarked$($itn$))}
         \State{$S$.push($itn$)}; \linecomment{// push the \enode}
         \EndIf
         \State{$itn$ $\gets$ $itn.\eleft$;} 
        \EndWhile
        \State{$itn$ $\gets$ $S$.pop()};
        \State{\vnode $adjn$ $\gets$ $itb$.\pointv;} 
        \If{($\rlx$($tid$,$cvn$.n,$adjn$,$itn$)} \label{lin:checknegcycle-rlx} 
         \State{return \tru;} //Neg-cycle present
         \EndIf
         \State{$itn$ $\gets$ $itn.\eright$;} 
        \EndWhile
    \EndFor
      \State{return \fal;}\linecomment{//Negative cycle is not present }
		\EndProcedure\label{checknegcycleend}
		\algstore{checknegcycle}
\end{algorithmic}
\hrule

\scriptsize
\begin{algorithmic}[1]
	\algrestore{checknegcycle}
		\scriptsize
		\renewcommand{\algorithmicprocedure}{\textbf{Method}}
		\Procedure{\Init($tid$)}{}\label{initstart} 
		\State{\hnode $hn$ $\gets$ \head;}
		\For{($i$ $\gets$ $0$; $i$ $<$ $hn.\sizee$; $i++$)} 
		\State{\fset $fsetnode$ $\gets$ $hn$.\buckets[$i$];}
		\State{$vn$ $\gets$ $fsetnode$.\vh;}
		\While{($vn$)}
		\If{($\neg$\isMarked($vn$))} 
		\State{$vn$.\opitem.\distarray[$tid$] $\gets$ $\infty$;}\linecomment{//Initialize \distarray[tid] }
		\EndIf
		\EndWhile
		\EndFor
  \EndProcedure\label{initend}
  \algstore{spinit}
\end{algorithmic}
\hrule

\scriptsize
\begin{algorithmic}[1]
	\algrestore{spinit}
		\scriptsize
		\renewcommand{\algorithmicprocedure}{\textbf{Method}}
		\Procedure{$\rlx$($tid$, $u$, $v$, $en$)}{}\label{rlxstart} 
         \If{($v$.\opitem.\distarray[$tid$] $>$ $u$.\opitem.\distarray[$tid$] $+$ $en$.\eweight)}
         \State{$v$.\opitem.\distarray[$tid$] $\gets$ $u$.\opitem.\distarray[$tid$] $+$ $en$.\eweight;} 
         \State{return \tru;} \linecomment{// relaxed}
         \Else 
         \State{return \fal;} \linecomment{// not relaxed}
         \EndIf
         \EndProcedure\label{rlxend}
        \algstore{sprlx}
\end{algorithmic}
\hrule
	\begin{algorithmic}[1]
	\algrestore{sprlx}
\renewcommand{\algorithmicprocedure}{\textbf{Method}}	
	\scriptsize
	\Procedure{$\sptclt$($\vkey$, $tid$)}{}\label{sptcltstart}
       \State{\spnode $spt$;} 
	   \State{queue$<$\bfsnode$>$$que$;}\linecomment{//Queue used for traversal}
	   \State{\Init($tid$);} \label{lin:sptclt-init}
	   \linecomment{//Initialize all vertices to $\infty$ distance from } 
         \State{\tcount $\gets$ \tcount $+$ 1;} \linecomment{//Thread local variable inc. by one}
         \State{$\vkey$.\opitem.\visitedarray[tid] $\gets$ \tcount;} 
         \linecomment{//Mark it visited} 
         \State{$\vkey$.\distarray[tid] $\gets$ 0 ;}  \label{lin:sptclt-init-vto0} \linecomment{//Set the dist of $\vkey$ to $0$}  
         \State{$spn$ $\gets$ new $\spnode$($\vkey$,\nul,\nul,$\vkey$.\opitem.\ecount);}
        \State{$spt$.$\addsp$($spn$);} 
        \linecomment{// add $spn$ to \sptree}
        \State{$que$.enque($spn$);} \linecomment{//Push the $spn$ into the $que$}
        \While{($\neg$ $que$.empty())} \linecomment{//Iterate over all vertices}
        \State{$cvn$ $\gets$ $que$.deque();} \linecomment{//Get the front node} 
       \If{($\neg$\isMarked($cvn$.n))}
       \State{\cntu;} \linecomment{//If marked then continue} 
        \EndIf
       \State{$itn$ $\gets$ $cvn$.n.\enext;} \linecomment{//Get the root \enode from its \elist}
        \State{stack $<$\enode $>$ $S$;} \linecomment{// stack for inorder traversal}
        \While{($itn$ $\vee$ $\neg S$.empty())}
        \While{($itn$ )}
         \If{($\neg$$\isMarked$($itn$))}
         \State{$S$.push($itn$)}; \linecomment{// push the \enode}
         \EndIf
         \State{$itn$ $\gets$ $itn.\eleft$;} 
        \EndWhile
        \State{$itn$ $\gets$ $S$.pop()};
        \If{($\neg$ $\isMarked$($itn$))} \linecomment{//Validate it}  
        \State{$adjn$ $\gets$ $itn$.\pointv;} \linecomment{//Get its \vnode} 
        \If{($\neg$$\isMarked$($adjn$))} 
        \If{($\rlx$($tid$,$cvn$.n,$adjn$,$itn$))} \label{lin:sptcltrexd} 
         
        \If{($\neg$$\checkvisited$($tid$,$adjn$,\tcount))} 
        \State{$adjn$.$\visitedarray$[$tid$]$\gets$\tcount;}
        \State{$spn$$\gets$new$\spnode$( $adjn$,$cvn$,\nul,} 
       \Statenolinnum{ $adjn$.\opitem.\ecount);}
        
        
         \State{$spt$.\addsp($spn$);} \label{lin:sptclt-addsp-spn} 
         \State{$que$.enque($spn$);} 
        \Else
       \State{\linecomment{/*If already visited and the distance}}
       \State{\linecomment{is relaxed further then update the}}
        \State{\linecomment{corresponding information in the $spt$*/}}
      \State{$spt$.$\updatesp$($tid$,$cvn$.$head$,$cvn$, $spn$);}\label{lin:sptclt-updatespt}
       \EndIf
        \EndIf
        \EndIf
        \EndIf
         \State{$itn$ $\gets$ $itn.\eright$;}
        \EndWhile
    \EndWhile
     \State{return $\langle spt, \checknegcycle(tid, spt)$;} \label{lin:sptclt-checknegcycle} 
		\EndProcedure\label{sptcltend}
		\algstore{bfssptclt}
\end{algorithmic}
\hrule
\begin{algorithmic}[1]
	\algrestore{bfssptclt}
		\scriptsize
		\renewcommand{\algorithmicprocedure}{\textbf{Method}}
		\Procedure{$\updatesp$($tid$, $head$, $par$, $spn$)}{}\label{updatespstart} 
        \State{\spnode $tn$ $\gets$ $head$;}
        \While{(tn)}
        \If{($tn.n.\vkey$ = $spn.\vkey$)}
        \State{$tn.n.\opitem.\distarray[tid]$ = $spn.\opitem.\distarray[tid]$;} \label{lin:updatesp-distarray} \linecomment{//update the shortest distance}
        \State{$tn.p$ = $par$;}\label{lin:updatesp-parent}
        \linecomment{//update the parent reference}
        \State{return;}
        \EndIf
        \State{$tn$ $\gets$ $tn.next$;}
        \EndWhile
         \EndProcedure\label{updatespend}
        \algstore{updatesp}
\end{algorithmic}
\end{multicols}
	\caption{The SSSP query.} \label{fig:sssp-methods}
\end{figure*}

%% file: code/bc-hash-code.tex
\begin{figure*}[!htp]
\begin{multicols}{2}

	\begin{algorithmic}[1]
	\algrestore{updatesp}
		\scriptsize
\renewcommand{\algorithmicprocedure}{\textbf{Operation}}
        \Procedure{$\getbc$($\vkey$)}{}\label{getbcstart}
        \State{$tid$$\gets$ \gettid();} \linecomment{//Get the thread id}
	    \If{($\isMarked$($\vkey$))}\linecomment{//Validate the vertex}
	     \State {return \nul; } \linecomment{//$\vkey$ is not present.}
        \EndIf
         \State{list $<\bfsnode>$  $st$; } 
         \State{$\langle s, st \rangle$ $\gets$ $\bcscan$($\vkey$, $tid$);}\linecomment{//Invoke the scan} \label{lin:getbfs-scan}
         \If{($s$ = \tru)} \label{lin:getbfs-bcstart}
         \State {$itn$ $\gets$ $st$.\enext; }
         \While{($itn$.\bnext $\neq$ \nul)} \linecomment{//Iterate all nodes}
          \State{$w$ $\gets$ $itn$.n;}
          \State{$bnode$$\gets$$w$.\opitem.$\predlist$[$tid$];}\linecomment{//Get the pred. list}
          \While{(\tru)} \linecomment{//Iterate all pred. nodes}
          \State{$vn$ $\gets$ $bnode$.n;}
          \State{\linecomment{//Compute the \deltaa value}}
         \State{$vn$.\opitem.$\deltaa$[$tid$] $\gets$ $vn$.\opitem.$\deltaa$[$tid$] +}
         \Statenolinnum{ $\frac{vn.\opitem.\sigmaa[tid]}{w.\opitem.\sigmaa[tid]}.(1+w.\opitem.\deltaa[tid])$;} \label{lin:dependcompute}
          \EndWhile         
         \If{($w$ $\neq$ $\vkey$)}
         \State{\linecomment{//Update the betweenness centrality}}
         \State{$w$.\opitem.$\betweenness$[$tid$] $\gets$ $w$.\opitem.$\betweenness$[$tid$] + $w$.\opitem.$\deltaa$[$tid$];}
         \EndIf
         \State{$itn$ $\gets$ $itn$.\bnext;}
          \EndWhile \label{lin:getbfs-bcend}
          \State{return $\vkey$.\betweenness[$tid$];}\linecomment{// return \betweenness of the vertex v}
         \EndIf 
		\EndProcedure\label{getbcend}
        \algstore{getbc}
\end{algorithmic}
\hrule

\begin{algorithmic}[1]
\algrestore{getbc}
	\scriptsize
	\renewcommand{\algorithmicprocedure}{\textbf{Method}}		
		\Procedure{$\bcscan$($\vkey$, $tid$)}{}\label{bcscanstart}
         \State{$\bfsnode$  $ot$, $nt$ ; } 
          \State{$ot$ $\gets$ $\bctclt$($\vkey$, $tid$);}\linecomment{//$1^{st}$ Collect} 
        \While{(\tru)} \label{bfswhilescan}\linecomment{//Repeat the tree collection}
           \State{$nt$ $\gets$ \bctclt($\vkey$, $tid$); } \linecomment{//$2^{nd}$ Collect}
           \If{($\comparetree$($ot$, $nt$))}\label{lin:bfsscan-comparetree}
           \State{return $\langle$\tru, $nt$ $\rangle$;} \linecomment{//Two collects are equal}
           \EndIf
           \State{$ot$ $\gets$ $nt$;}\label{bfsrescan} \linecomment{//Retry if two collects are not equal}
    \EndWhile
		\EndProcedure\label{bcscanend}
        \algstore{bcscan}
\end{algorithmic}
\hrule
\scriptsize
\begin{algorithmic}[1]
\algrestore{bcscan}
\renewcommand{\algorithmicprocedure}{\textbf{Method}}
	\Procedure{$\bctclt$($\vkey$, $tid$)}{}\label{bctcltstart}
        \State{\bfsnode $bt$;}  
	   \State{queue$<$\bfsnode$>$$que$;} 
         \State{\tcount$\gets$\tcount$+$1;} \linecomment{//Thread local variable inc. by one}
         \State{$\vkey$.\opitem.\visitedarray[$tid$] $\gets$ \tcount ;} \linecomment{//Mark it visited} 
         \State{$bn$ $\gets$ new $\bfsnode$($\vkey$, \nul, \nul, $u$.\ecount);}
        
         \State{$que$.enque($bn$);}\linecomment{//Push the initial node}
        \While{($\neg$ $que$.empty())}\linecomment{//Iterate over all nodes}
        \State{$cvn$ $\gets$ $que$.deque();} \linecomment{// Get the front node}
       \If{($\neg$\isMarked($cvn$.n))}  \linecomment{// check the node is marked or not}
       \State{\cntu;} \linecomment{// If marked then continue} 
        \EndIf 
        \State{$bt$.$\add$($cvn$);} \label{lin:bctclt-add1}\linecomment{//Insert $cvn$ to the \bfstree $bt$}
        \State{$itn$ $\gets$ $cvn$.n.\enext;} 
        \State{stack $<$\enode $>$ $S$;} \linecomment{// stack for inorder traversal}
        \While{($itn$ $\vee$ $\neg S$.empty())}
        \While{($itn$ )}
         \If{($\neg$$\isMarked$($itn$))}
         \State{$S$.push($itn$)}; \linecomment{// push the \enode}
         \EndIf
         \State{$itn$ $\gets$ $itn.\eleft$;} 
        \EndWhile
        \State{$itn$ $\gets$ $S$.pop()};
        \If{($\neg$$\isMarked$($itn$))} 
        \State{$adjn$ $\gets$ $itn$.\pointv;}
        \If{($\neg$$\isMarked$($adjn$))} 
        \If{($adjn$.\opitem.\distarray[$tid$] = $\infty$)} \label{lin:bctclt-pathdescover1}
        \State{\linecomment{//Update the distance} }
        \State{$adjn$.\opitem.$\distarray$[$tid$]$\gets$$cvn$.$n.\opitem.\distarray$[$tid$]+1;}\label{lin:bctclt-pathdescover2}
        \State{$bn$$\gets$new$\bfsnode$($adjn$,$cvn$,\nul,} 
       \Statenolinnum{ $adjn$.\ecount);}
        \State{$que$.enque($bn$);}\linecomment{//Push the node into que }
        \EndIf
         \State{\linecomment{//path counting}}
        \If{($adjn$.\opitem.$\distarray$[$tid$] = $cvn$.$n$.\opitem.$\distarray$[$tid$] +1)}  
        \label{lin:bctclt-pathcount1}
        \State{\linecomment{//Set the \sigmaa value}}
        \State{$adjn$.\opitem.$\sigmaa$[$tid$]$\gets$$adjn$.\opitem.$\sigmaa$[$tid$]+ } 
        \Statenolinnum{ $cvn$.n.\opitem.$\distarray$[$tid$]+1;} \label{lin:bctclt-sigupdate}
        \State{$pn$$\gets$new $\pnode$(cvn.n,\nul,\nul);}
        \State{\linecomment{//Insert $pn$ into the \predlist of $cvn$}}
        \State{$cvn$.$n.\opitem.\predlist$[$tid$].$\add$($pn$);} \label{lin:bctclt-predlist}
        \EndIf
        \EndIf
        \EndIf
         \State{$itn$ $\gets$ $itn.\eright$;}
        \EndWhile
    \EndWhile
     \State{return $bt$;} \linecomment{//Return the tree to the scan method}
		\EndProcedure\label{bctcltend}
\end{algorithmic}
\end{multicols}

	
\caption{The BC query.} \label{fig:bc-methods}
\end{figure*}

%% file: Graph-Applications/bc.tex
We defined Betweenness centrality (BC) in Section \ref{sec:intro}. BC is an index measure based on the relative significance of a vertex node in a graph \cite{Freeman:bc:1977}. 
Here we consider unweighted graphs $G=(V,E)$: $w_e=1~\forall e\in E$.

We formally define the BC as the following. Given a directed graph $G=(V, E)$ and some $s, t \in V$, let $\sigma(s,t)$ be the number of shortest paths between vertex $s$ and $t$ and $\sigma(s, t|v)$ be the number of shortest paths between $s$ and $t$ that pass through an intermediate vertex $v$. Then the \textit{pair-dependency} of $s, t$ on $v$ is defined as $\delta(s,t|v) = \frac{\sigma(s, t|v)}{\sigma(s,t)}$, if $s = t$, then $\sigma(s,t)=1$ and if $v \in s, t $, then $\sigma(s,t|v)=0$. $\sigma(s,t)$ can be compute recursively as $\sigma(s,t) = \sum_{u\in Pred(t)} \sigma(s,u)$, where $Pred(t) = \{u : (u, t) \in E, d(s, t) = d(s, u) + 1\}$ (predecessors of $t$ on shortest path from $s$), and $d(s, u)$ is the distance between vertex node $s$ and $u$.  With that, the BC of $v$ is defined as $C_B(v) = \sum_{s, t \in V}\delta(s,t|v)$. To compute $\delta(s,t|v)$ one can run the BFS algorithm with each vertex node as source $s$ and then sum the pair-dependencies for each $v\in V$. 

U. Brandes \cite{Brandes:BC:JMS:2001} proposed an algorithm for unweighted graphs by defining an \textit{one-sided dependencies} equation as $\delta(s|v)$ = $\sum_{(v,w)\in E, w: d(s,w)=d(s,v)+1}$ $\frac{\sigma(s,v)}{\sigma(s,w)}\times (1+\delta(s|w))$, and the BC be $C_B(v) = \sum_{s \in V}\delta(s|v)$.
Brande's algorithm is as follows: it iterates over the vertices $s \in V$ and then computes $\delta(s|v)$ for $v \in V$ in two phases, (1) using the BFS algorithm, it computes the distances and shortest path counts from $s$, and it also keeps tracking of all the vertices onto a stack as they are visited, (2) it visits all the vertices by popping them off from the stack in reverse order and aggregate the dependencies according to the one-sided dependencies equation. 

In our setting, the operation $\getbc\in\mathcal{Q}$, Lines \ref{getbcstart} to \ref{getbcend} in Figure \ref{fig:bc-methods}, builds on \getbfs. In this case, the \opstruct class, in addition to \ecount and \visitedarray, contains five extra array fields: \distarray, \sigmaa, \deltaa, \predlist and \betweenness, whose cells correspond to the threads.
In essence, we adapt the algorithm of Brandes \cite{Brandes:BC:JMS:2001} to our setting using these arrays to allow concurrent threads compute the measure of BC anchored at their corresponding array-cells.

The termination criterion of the method \bcscan is similar to \scan: matching a consecutive pair of specialized partial snapshots, called \textit{BC-tree}. Besides that, the snapshot method \bctclt extends \treeclt as the following:
\begin{enumerate}[leftmargin=0.5cm,align=left,labelwidth=\parindent,labelsep=0pt,topsep=0pt,itemsep=0pt]
	\item[\textsc{\textbf{Distance tracking:~}}] On visiting a vertex, set its distance from \vkey at the corresponding \distarray cell, Line \ref{lin:bctclt-pathdescover2}.
	\item[\textsc{\textbf{Shortest path counting:~}}] If a visited vertex is on the shortest path, adjust the number of shortest paths recorded at the corresponding \sigmaa cell, Line \ref{lin:bctclt-sigupdate}.
	\item[\textsc{\textbf{Predecessor list maintenance:~}}] For each of the visited vertices record and update the list of predecessors from \vkey by the shortest path, Line \ref{lin:bctclt-predlist}. 
\end{enumerate} 
On termination of a \bcscan, see Lines \ref{lin:getbfs-bcstart} to \ref{getbcend}, the final snapshot is processed as the following: for each of the vertices $u$ pointed to by the \treenode{s} of the collected BC-tree, back-propagate to the source $\vkey$ along the predecessor-list; recursively aggregate the \textit{dependencies} w.r.t. the predecessors in $\deltaa[tid]$, where $tid$ is the thread id, Line \ref{lin:dependcompute}; accumulate the dependencies to compute and store \bc in $\betweenness[tid]$, Line \ref{lin:dependcompute}. For space limitation, we skip the detail description of these terminologies here and refer a reader to \cite{Brandes:BC:JMS:2001}.



%% file: results.tex
Evaluating the performance of a new concurrent data structure vis-a-vis the existing ones implementing the same ADT, in general, limits to comparing the throughput or latency of a combined execution of all of the proposed ADT operations. Moreover, in general, they use non-standard synthetic micro-benchmarks. In contrast, the high performance static graph-query libraries are evaluated for the average latency of a parallelized execution of a single graph operation. Furthermore, the datasets used in the evaluation of graph-query libraries are now standard in the literature. As we aimed to present this work as a library of graph operations in a dynamic setting, keeping the aforementioned points in view, we compared the experimental performance of our \nbk graph algorithms against a well-known graph-query library \textbf{Ligra}~\cite{Shun+:ligra:ppopp:2013}. The dynamic updates for Ligra are simulated by intermittent sequential addition and removal in the dataset. We use a standard synthetic graph dataset -- \textbf{R-MAT} graphs \cite{Chakrabarti+:RMAT:SIAM:2004} -- with power-law distribution of degrees. 

As we discussed before, we needed the repeated snapshot collection and matching methodology to guarantee \lbty of graph queries. However, if the consistency requirement is not as strong as \lbty, we can still have \nbk progress even if we collect the snapshots once. At the cost of theoretical consistency, we gain a lot in terms of throughput, which is the primary demand of the analytics applications, who often go for approximate queries. Thus, we have the following execution cases: 
(1) \textbf{PG-Cn}: Linearizable \panighm, (2) \textbf{PG-Icn}: Inconsistent \panighm, and (3) \textbf{\Ligra}: execution on Ligra\cite{Shun+:ligra:ppopp:2013}.

\begin{figure}[!htb]
    \flushleft
    \includegraphics[scale=0.82]{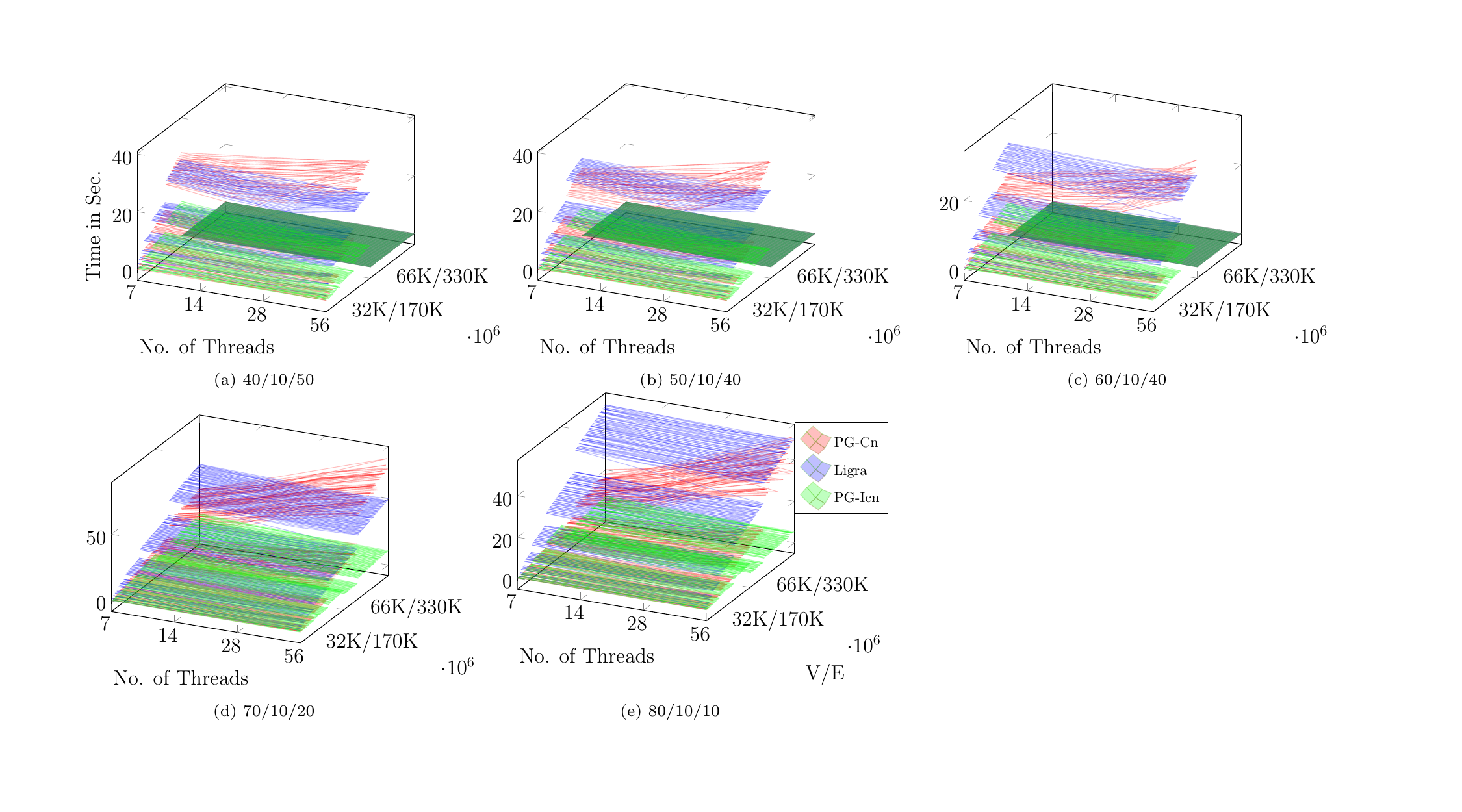}
    \vspace{-3mm}
    \caption{ Latency of the executions containing \getgphalgo: \getbfs. The surfaces indicate the total time for an end-to-end run of $10^4$ operations uniformly distributed by a distribution mentioned below the plots. x-axis refers to the number of threads, whereas, y-axis refers to the ratio of number of vertices and edges for a graph instance.}
    \label{fig:bfsgraph}
\end{figure}

\begin{figure}[!htb]
    \flushleft
    \includegraphics[scale=0.82]{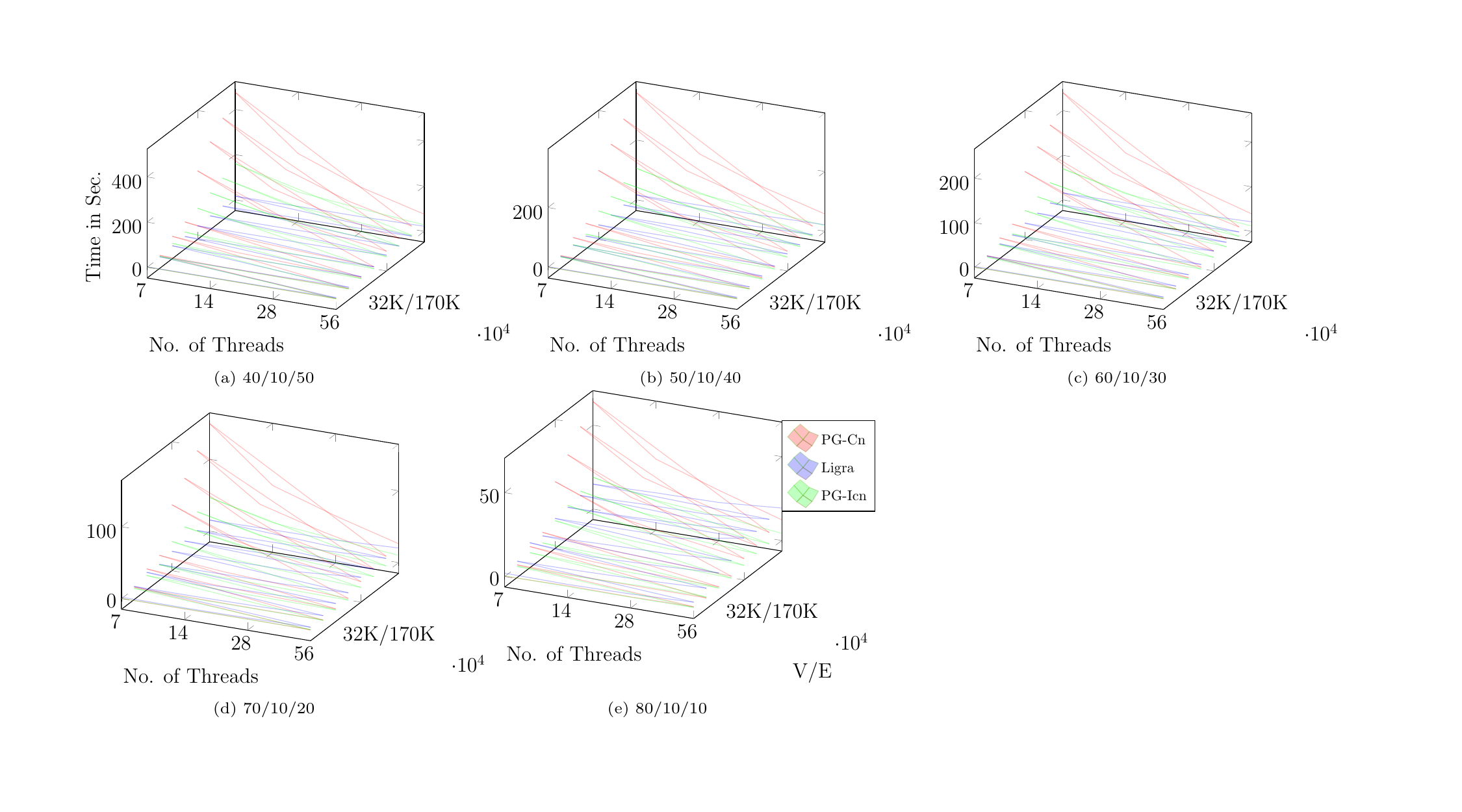}
    \vspace{-3mm}
    \caption{   Latency of the executions containing \getgphalgo: SSSP.}
\label{fig:ssspgraph}
\end{figure}
\begin{figure}[!htb]
    \flushleft
    \includegraphics[scale=0.82]{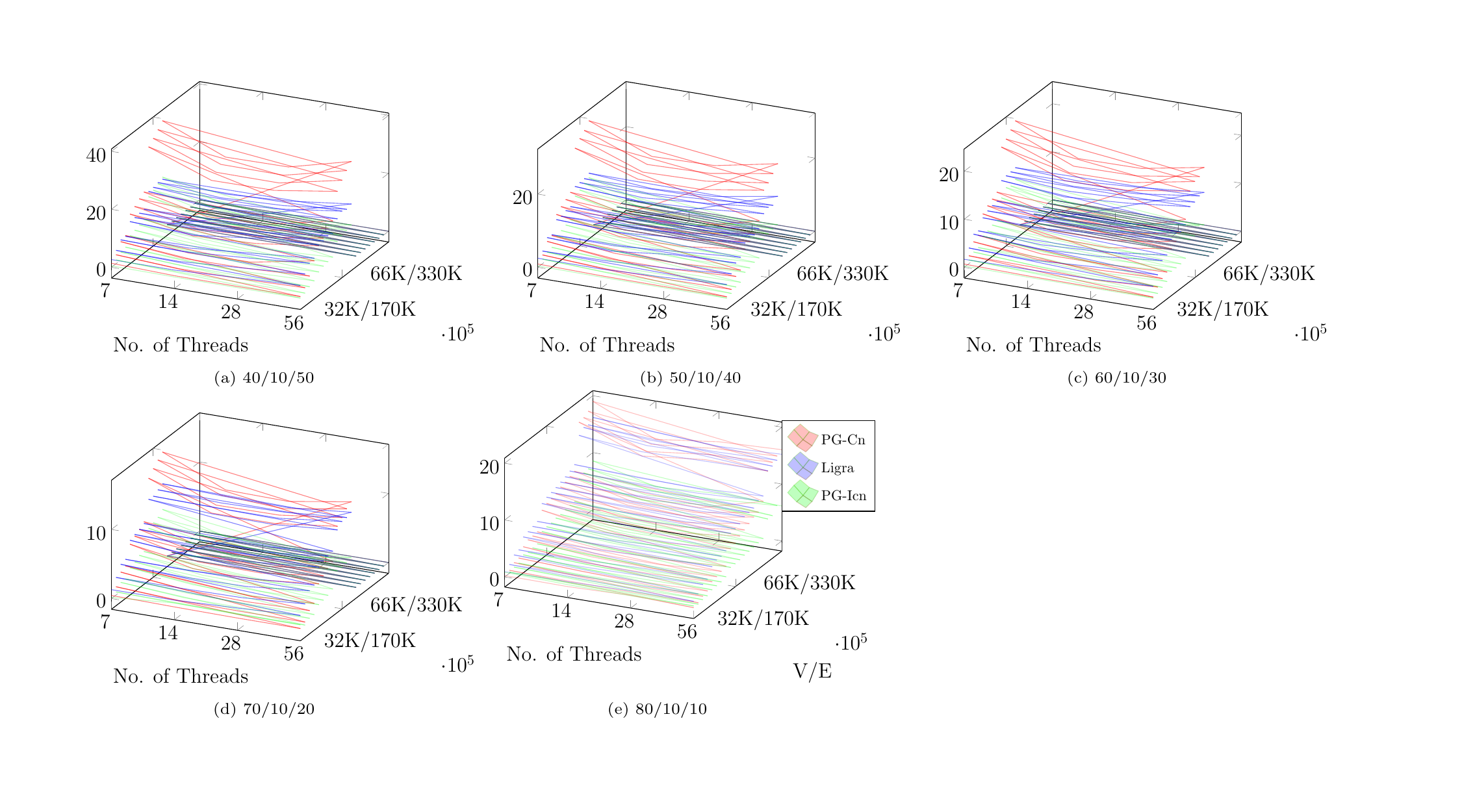}
    \vspace{-3mm}
\caption{  Latency of the executions containing \getgphalgo: \getbc.}
\label{fig:bcgraph}
\end{figure}

\textbf{R-MAT Graph Generation:}
R-MAT graph was described by Chakrabarti, et al. \cite{Chakrabarti+:RMAT:SIAM:2004}. It is a recursive matrix (R-MAT) model that can quickly generate realistic graphs using very few input parameters. It utilizes power-law/DGX \cite{Zhiqiang+:DGX:KDD:2001} degree distributions, which match the characteristics of real-world graphs (degree exponents, diameters, etc.). The R-MAT model generates the graphs in $O(E log(E) log(N))$ time (Explanations can be seen in \cite{Chakrabarti+:RMAT:SIAM:2004}). R-MAT model uses adjacency matrix $A$ of a graph of $N$ vertices and an $N * N$ matrix, with entry $A(i, j) = 1$ if the edge $(i, j)$ exists, and $0$ otherwise. It recursively subdivides the adjacency matrix into four equal-sized partitions and distributes edges within these partitions with unequal probabilities $a$, $b$, $c$, and $d$ following a certain distribution. For our case, the default values are $a=0.5, b=0.1, c=0.1$ and $d=0.3$, such that $a + b + c + d = 1$, and these values can be changed based on the required application. 

To generate a graph, the input parameters required at the running time are the number of vertices and the output file name. The default number of edges is set to $10$ times the number of vertices, and that is changed depending on the required application by setting the flag followed by the number of edges. Table \ref{tab:rmat-graphs} shows the different graphs used in the paper for running the BFS, SSSP, and BC algorithms. 

To generate a weighted graph, we added random integer weights in the range [$1$, $\ldots$, $log_2(N)$] to an unweighted graph in adjacency graph format. 

\begin{table}[h!]
	\centering
	\begin{tabular}{ |p{2cm}|p{7cm}|  }
		\hline
		$Vertices$ & $Edges$\\
		\hline
		\hline
		1024 & 10000 \\
		2048 & 20000 \\
		4096 & 30000, 40000 \\
		8192 & 50000 , 60000, 70000, 80000\\
		16384 & 90000 , 100000, $\ldots$, 160000\\
		32768 & 170000, 180000, $\ldots$, 320000\\
		65536 & 330000, 340000, $\ldots$, 650000 \\
		131072 & 660000, 670000, $\ldots$, 1000000 \\ [1ex] 
		\hline
	\end{tabular}
	\caption{Initial RMAT graphs used to load before start running BFS, SSSP, and BC algorithms and then perform $10^4$ operations with different workload distributions.}
	\label{tab:rmat-graphs}
\end{table}


\begin{figure}[!htb]
    \flushleft
    \includegraphics[scale=0.82]{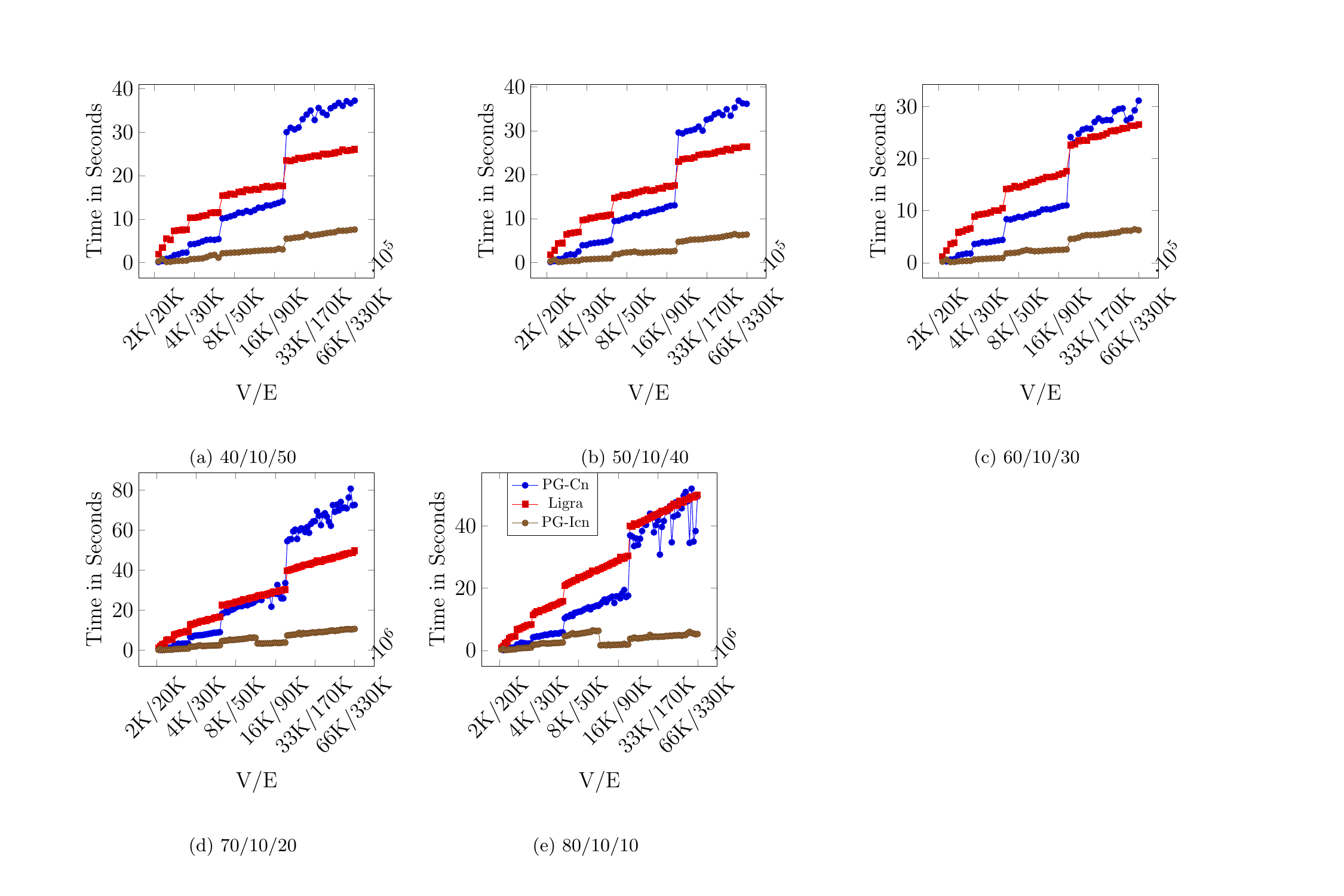}
    \vspace{-3mm}
     \caption{\small  End-to-end time for an execution of $10^4$ operations containing \getgphalgo: \getbfs with 56 threads. The x-axis refers to the ratio of number of vertices and number of edges for a graph instance. The distribution of operations are mentioned below each figure.}
    \label{fig:bfs-56t-2d}
\end{figure}

\begin{figure}[!htb]
    \flushleft
    \includegraphics[scale=0.82]{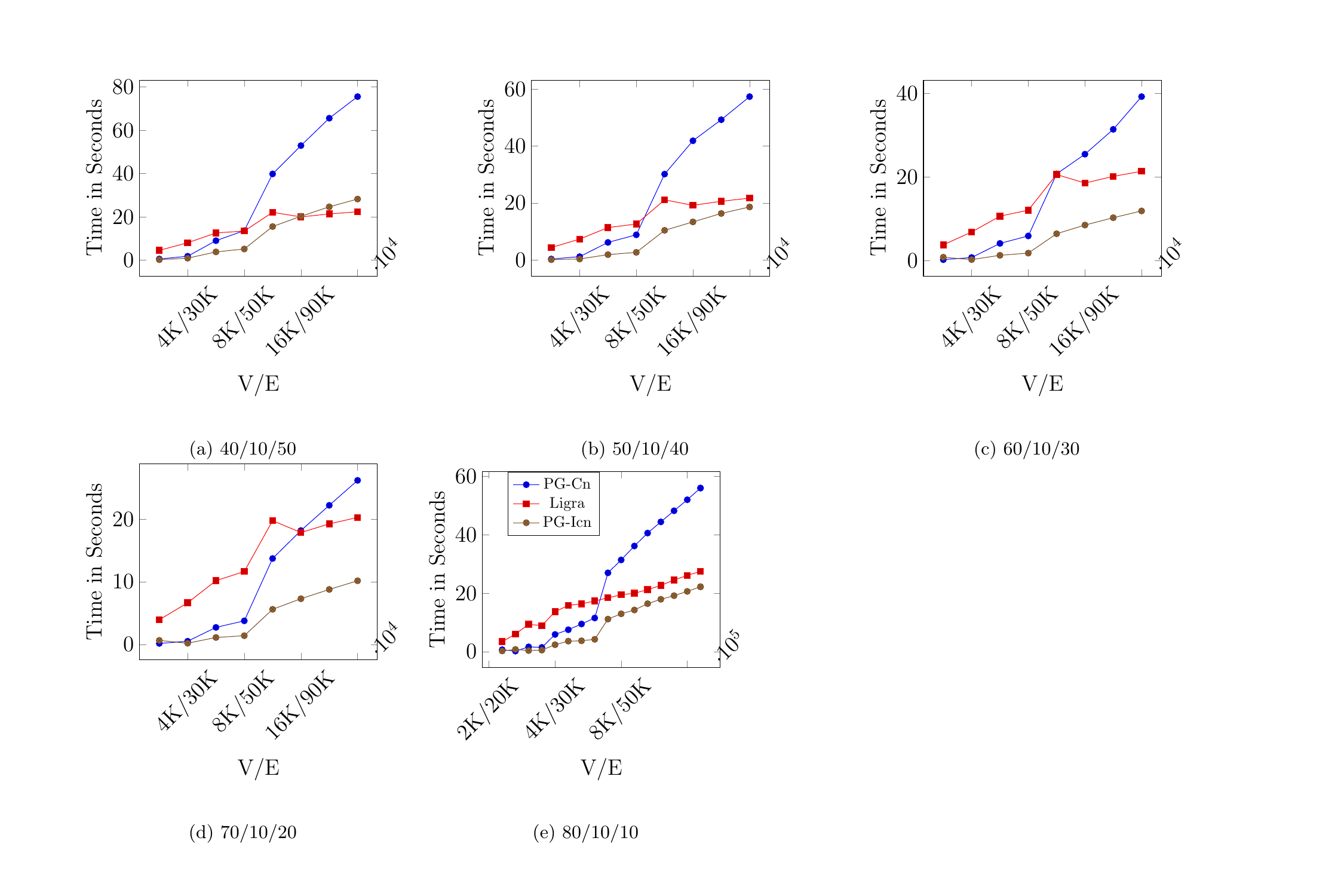}
    \vspace{-3mm}
     \caption{   Execution time for $10^4$ operations containing \getgphalgo: \getsp with 56 threads.}
    \label{fig:SSSP-56t-2d}
\end{figure}

\begin{figure}[!htb]
    \flushleft
    \includegraphics[scale=0.82]{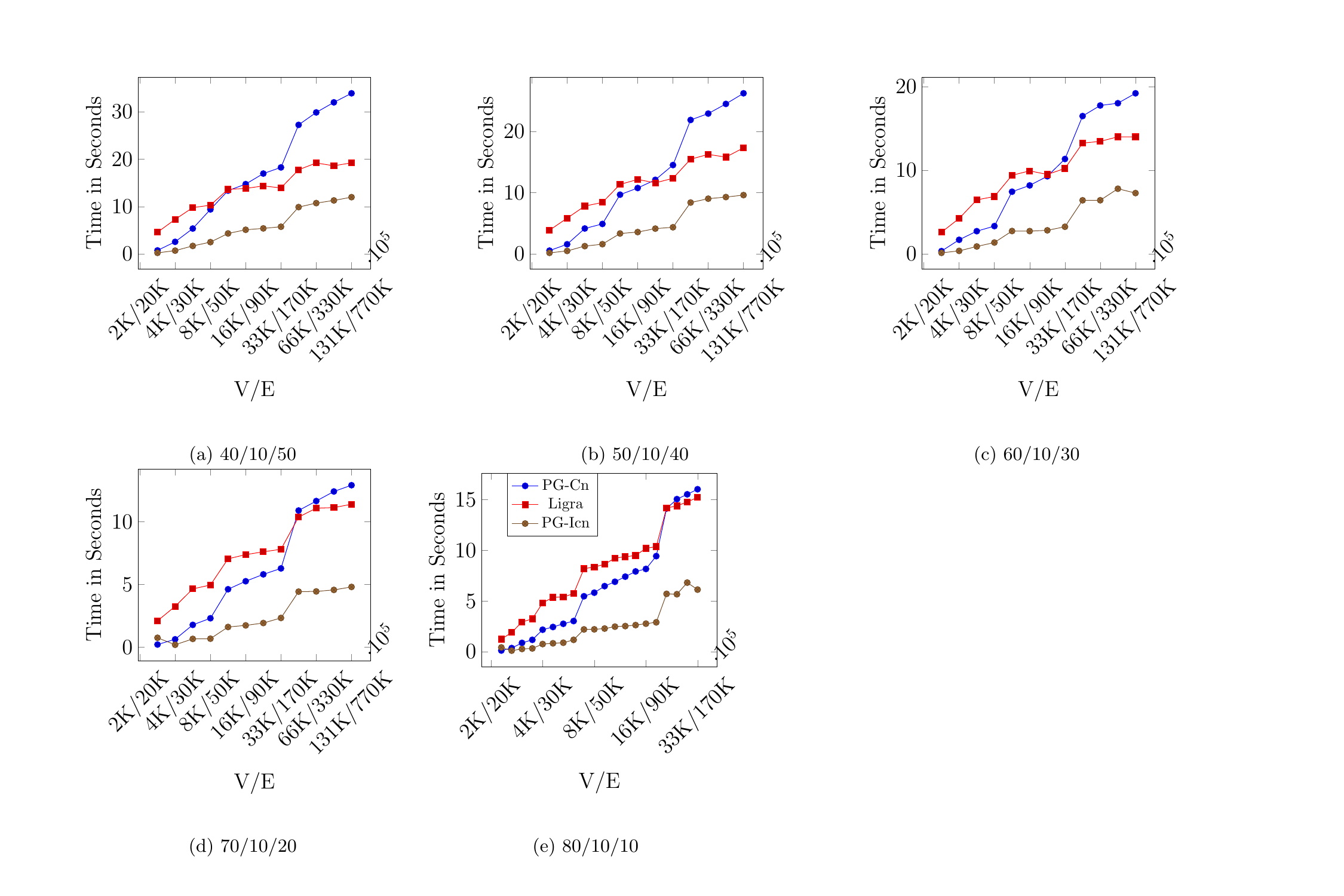}
    \vspace{-3mm}
     \caption{ Execution time for $10^4$ operations containing \getgphalgo: \getbc with 56 threads.}
    \label{fig:BC-56t-2d}
\end{figure}

\begin{figure}[!htb]
    \flushleft
    \includegraphics[scale=0.82]{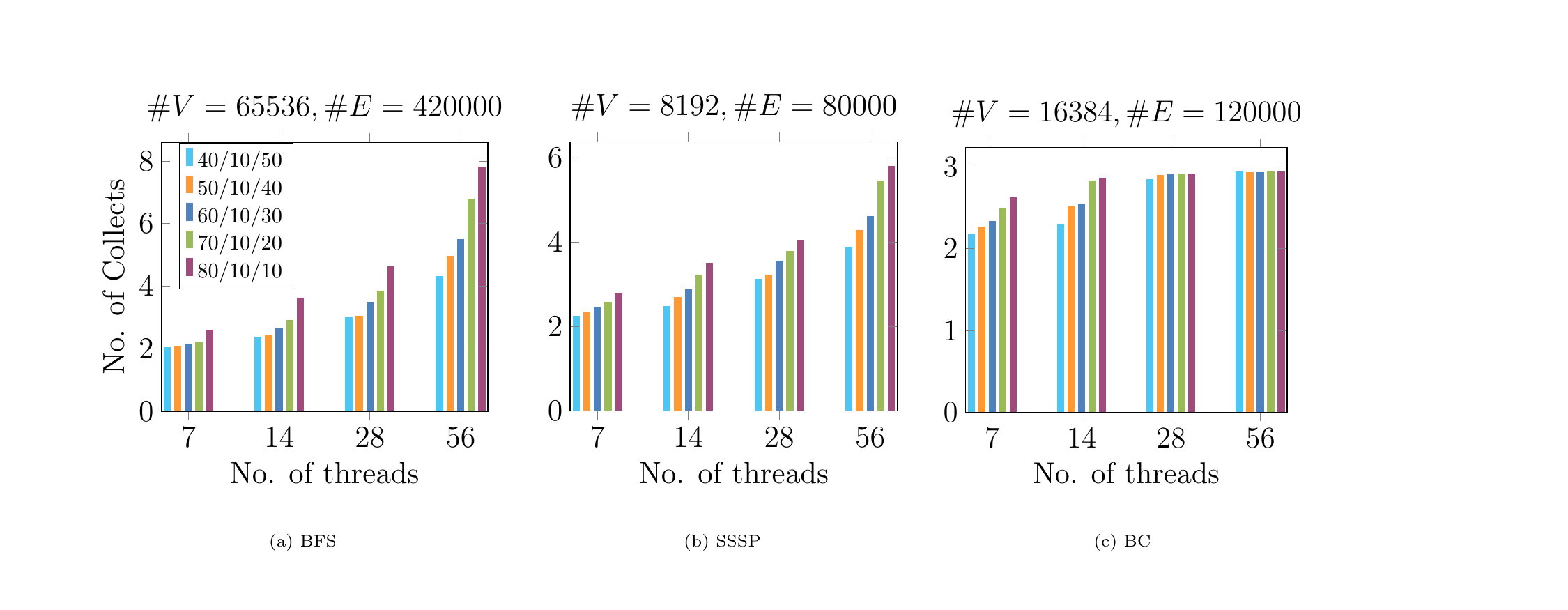}
    \vspace{-3mm}
    \caption{The average number of \collect operations for each \scan for executions of PG-Cn in Figures \ref{fig:bfsgraph}, \ref{fig:ssspgraph}, and \ref{fig:bcgraph} are plotted in Subfigs. (a), (b), and (c), respectively.}
    \label{fig:Scan-Updates}
\end{figure}

\begin{figure}[!htb]
    \flushleft
    \includegraphics[scale=0.82]{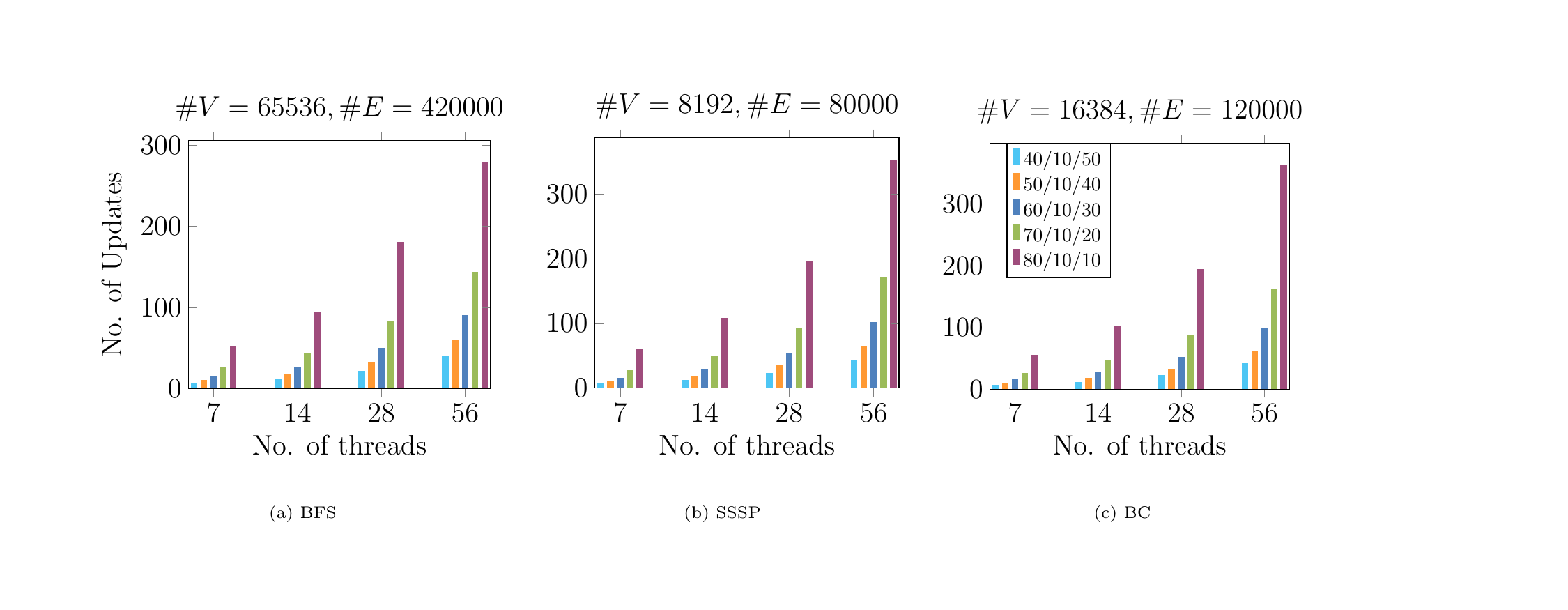}
    \vspace{-3mm}
     \caption{Subfigs. (a), (b), and (c) show the average number of interrupting update operations during the lifetime of the respective queries.}
    \label{fig:Scan-Updates2}
\end{figure}

\textbf{Experimental Setup:}
We conducted our experiments on a system with Intel(R) Xeon(R) E5-2690 v4 CPU packing 56 cores running at 2.60GHz. There are 2 logical threads for each core and each having private cache memory L1-64K and L2-256K. The L3-35840K cache is shared across the cores. The system has  32GB of RAM and 1TB of hard disk. It runs on a 64-bit Linux operating system. 
All the implementations \ignore{\footnote{The source code is available on https://github.com/sngraha/PANIGRAHAM.}} are written in C++ without garbage collection. We used Posix threads for  multi-threaded implementation. 

\textbf{Running Strategy:}
The experiments start with a R-MAT graph instance populating the data structure. At the execution initialization, we spawn a fixed set of threads (7, 14, 28 and 56) and each thread randomly performs a set of operations chosen by a certain random workload distribution. The metric for evaluation is the total time taken to complete the set of operations, after a fixed warm-up: $5\%$ of the total number of operations. Each experiment runs for $5$ iterations and then we take the median of all iterations.

\textbf{Workload Distribution:}
 To evaluate the performance over a number of micro-benchmarks, we used a range of distributions over an ordered (family of) set of operations: $\{Update{:=}\{\addv,  \remv, \adde, \reme\},Search{:=}\{\conv,\linebreak \cone\},\getgphalgo)\}$. In each case, first we load a RMAT graph instance, perform warm-up operations, followed by an end-to-end run of $10^4$ operations in total, assigned in a uniform random order to the concurrent threads. In the plots a label, say, $40/10/50$ refers to a distribution $\{\{10\%,10\%,10\%,10\%\},\{5\%,5\%\},50\%\}$ (i.e. $40\%$ first four operations distributed equally, $10\%$ next two distributed equally, and $50\%$ \getgphalgo queries) of the aforementioned operations in their order. 

 \subsection*{Observations and Discussion}
 Firstly, the surface plots of latency, in Figures \ref{fig:bfsgraph}, \ref{fig:ssspgraph}, and \ref{fig:bcgraph}, show that for any combination of a graph size and parallelism, as in the number of threads, the inconsistent \nbk implementation outperforms higly parallel Ligra by an order of magnitude. It depicts the advantage of concurrency in the dynamic settings. The surfaces also show that for \getbfs, where the \nbk methods need to maintain only a small amount of information in vertices, the linearizable \nbk concurrent executions  outperform Ligra up to a reasonable graph size. In case of \getsp and \getbc, where a lot more info is needed for linearizable snapshots, Ligra tends to work better than consistent \nbk implementation, still under-performing the one that relaxes consistency.
 
 Secondly, for a fixed number of threads, in this instance 56, as shown in  Figures \ref{fig:bfs-56t-2d}, \ref{fig:SSSP-56t-2d}, and \ref{fig:BC-56t-2d}, for the smaller graph sizes, \nbk implementations handsomely outperform Ligra, however, as graph size increases, Ligra starts getting advantage of the parallel implementation. Note that, the \nbk implementations do not have even an inline parallelization.  
 
 To explore the effect of concurrent dynamic updates on linearizable query performance of PG-Cn, we plot the average number of snapshot collection method calls in \figref{Scan-Updates} (a), (b), and (c), and the average number of interrupting updates during the lifetime of a query in  \figref{Scan-Updates2} (a), (b), and (c).
 We observe that the average number of snapshot collection grows in a linear ratio to the number of threads, which is on the expected lines. However, comparing the average number of snapshot collections and the corresponding average number of interrupting updates for the respective queries, we observe that the former does not grow in the same ratio as the latter when the ratio of updates in a distribution increases. This shows that a smaller number of queries in a distribution results in a lesser amount of interaction between a query and the concurrent updates. This observation infers that consistent \nbk implementation of dynamic graph queries is a much better option than sequential highly parallel implementation, such as Ligra, where by design a query has to be stopped for an update.

The experimental observations infer that a design which can take advantage of both concurrency and parallelization will substantially benefit the dynamic graph queries. We plan to work on this in future.
 
 

%% file: appendix.tex
\begin{center}
\Large\textbf{Appendix}
\end{center}
\label{sec:appendix}



\section{The \Nbk Graph Algorithm} \label{sec:apdx-operations}
\input{apdx-graphOperations.tex}

\section{Proof of Correctness and Progress Guarantee}\label{sec:proof}
\input{proof}





%% file: apdx-graphOperations.tex
\input{code/hash-bst-graph-code1.tex}

In this section, we present a detailed implementation of our \nbk directed 
graph algorithm. The \nbk graph composes on the basic structures of the
dynamic \nbk hash table \cite{Liu+:LFHash:PODC:2014} and \nbk internal 
binary search tree \cite{Howley+:NbkBST:SPAA:2012}. For a self-contained 
reading, we present the 
algorithms of \nbk hash table and BST. Because it derives and 
builds on the earlier works \cite{Liu+:LFHash:PODC:2014} and 
\cite{Howley+:NbkBST:SPAA:2012}, many keywords in our presentation are identical to 
theirs. One key difference between our \nbk BST design from 
\cite{Howley+:NbkBST:SPAA:2012} is that we maintain a mutable edge-weight in 
each BST node, thereby not only the implementation requires extra steps but 
also we need to discuss extra cases in order to argue the correctness of our 
design. Furthermore, we also perform non-recursive traversals in the BST for 
snapshot collections, which were already discussed as part of the graph 
queries. The pseudo-codes pertaining to the \nbk hash-table are presented in 
\figref{nbk-hash-methods}, whereas those for the \nbk BST are presented in 
Figures \ref{fig:nbk-bst-method1} and \ref{fig:nbk-bst-method2}.

\subsection{Structures}
The declarations of the object structures that we use to build the data structure are listed in 
\figref{nbk-hash-methods} and \ref{fig:nbk-bst-method1}. The structures \fset, 
\fsetop, and \hnode are used to build the \textit{\vlist}, whereas \Node, 
\relocateop, and \childcasop are the component-objects of  the \textit{\elist}. The 
structure \fset, a freezable set of \vnodes that serves as a building block of 
the \nbk hash table. An \fset object builds a \vnode set with \addv, \remv and 
\conv operations, and in addition, provides a \freeze method that makes the 
object immutable. The changes of an \fset object can be either addition or 
removal of a \vnode. For simplicity, we encode \addv and \remv operation as 
\fsetop objects. The \fsetop has a state \optype (\addv or \remv), the \key 
value, \done a boolean field that shows the operation was applied or not, and 
\resp a boolean field that holds the return value.

The \vlist is a dynamically resizable \nbk hash table constructed with the 
instances of 
\vnodes, and it is a linked-list of \hnodes(Hash Table Node). The \hnode 
composed of an array of buckets of \fset objects, the \Hsize field stores the 
array length and the predecessor \hnode is pointed to by the \hpred pointer. 
The head of the \hnode is pointed to by a shared \head pointer.

For clarity, we assume that a \resize method grows (doubles) or shrinks 
(halves) the \Hsize of the \hnode which amount to modifying the length of the 
bucket array. The hash function uses modular arithmetic for indexing in the 
hash table, e.g. \idx = \key $\Mod$ \Hsize.   

Based on the boolean parameter taken by \resize method, it decides the hash table either to grow or shrink. The \initbucket method ensures all \vnodes are physically present in the buckets. It relocates the \hnodes to the hash table which are in the predecessor's list. 

The $i^{th}$ bucket of a given \hnode $h$ is initialized by \initbucket method, 
by splitting or merging the buckets of $h'$s predecessor \hnode $s$, if $s$ 
exists. The sizes of $h$ and $s$ are compared and then this method decides 
whether $h$ is shrinking or expanding with reference to $s$. Then it freezes 
the respective bucket(s) of $s$ before copying the \vnodes. If $h$ halves the 
size of $s$, then $i^{th}$ and $(i + h.\Hsize)^{th}$ buckets of $s$ are merged 
together to form the $i^{th}$ bucket of $h$. Otherwise, $h$ doubles the size of 
$s$, then approximately half of the \vnodes in the ($i$ \Mod $h.\Hsize)^{th}$ 
bucket of $s$  	relocate to the $i^{th}$ bucket of $h$. To avoid any races 
with the other helping threads while splitting or merging of buckets a \CAS is 
used (Line \ref{lin:resize-cas}).

The \enode structure is similar to that of a \lf 
BST~\cite{Howley+:NbkBST:SPAA:2012} with an additional edge weight \eweight and 
a pointer field \pointv which points to the corresponding \vnode. This helps 
direct access to its \vnode while doing a BFS traversal and also helps in 
deletion of the incoming edges. The operation \oper field stores if any changes 
are 
being made, which affects the \enode. To avoid the overhead of another field in 
the node structure, we use bit-manipulation: last significant bits of a 
pointer $p$, which are unused because of the memory-alignment of the 
shared-memory system, are used to store information about the state of the 
pointer shared by concurrent threads and executing an operation that would 
potentially update the pointee of the pointer. More specifically, in case of an 
x86\text{-}64 bit architecture, memory has a 64-bit boundary and the last three 
least significant bits are unused. So, we use the last two significant bits, 
which are enough for our purpose, of the pointer to store auxiliary data. We 
define four different methods to change an \enode pointer: 
$\isnull(p)$ returns \tru if the last two significant bits of $p$ make 
$00$, which indicates no ongoing operation, otherwise, it returns \fal; 
$\isMarked(p)$ returns \tru if the last two significant bits of $p$ are set to 
$01$, else it returns \fal, which indicates the node is no longer in the tree 
and it should be \textit{physically} deleted; $\ischildcas(p)$ returns \tru if 
last two bits of $p$ are 
set to $10$, which indicates one of the child node is being modified, else it 
returns \fal;  $\isrelocate(p)$ returns \tru if the last two bits of  
$p$ make $11$, which indicates that the \enode is undergoing a node relocation 
operation. 

A \childcasop object holds sufficient information for another thread to finish an operation that made changes to one of the child -- right or left -- pointers of a node. A node's \oper field holds a flag indicating an active \childcasop operation. Similarly, a \relocateop object holds sufficient information for another thread to finish an operation that removes the key of a node with both the children and replaces it with the next largest key. To replace the next largest key, we need the pointer to the node whose key is to be removed, the data stored in the node's \oper field, the key to replacement and the key being removed. As we did in case of a \childcasop, the \oper field of a node holds a flag with a \relocate state indicating an active \relocateop operation.

\subsection{The Vertex Operations}
\label{subsec:nbk-vops}
The working of the \nbk vertex operations $\addv$, $\remv$, and $\conv$ are 
presented in \figref{nbk-graph-methods}. A $\addv(\vkey)$ operation, at Lines 
\ref{addvstart} to \ref{addvend}, invokes $\hashadd(\vkey)$ to perform an 
insertion of a \vnode $\vkey$ in the hash table. A $\remv(\vkey)$ operation at 
lines \ref{remvstart} to \ref{remvend} invokes $\hashrem(\vkey)$ to perform 
a deletion of \vnode $\vkey$ from the hash table. The method \apply, which tries to modify the corresponding buckets, is called by both $\hashadd$ and 
$\hashrem$, see Line \ref{lin:addv-apply} and \ref{lin:remv-apply}. It first creates a new \fsetop object consisting of the 
modification request, and then constantly tries to apply the request to the respective bucket $b$, see Lines \ref{applystart} to 
\ref{applyend}. Before applying the changes to the bucket it checks 
whether $b$ is \nul; if it is, \initbucket method is invoked to 
initialize the bucket (Line \ref{lin:apply-initbcucket}). At the end, the return 
value is stored in the \resp field.

The algorithm and the resizing hash table are orthogonal to each other, so we 
used heuristic policies to resize the hash table. As a classical heuristic we 
use a \hashadd operation that checks for the size of the hash table with some 
\textit{threshold} value, if it exceeds the threshold the size of the table is 
doubled. Similarly, a \hashrem checks the threshold value, if it falls below 
threshold, it shrinks the hash table size to halves.   

A $\conv(\vkey)$ operation, at Lines \ref{convstart} to \ref{convend}, invokes 
$\hashcon(\vkey)$ to search a \vnode $\vkey$ in the hash table. It starts by 
searching the given key $\vkey$ in the bucket $b$. If $b$ is \nul, it reads 
$t'$s predecessor (\linref{hashcon-conv-pred}) $s$ and then starts searching on 
it. At this point it could return an incorrect result as \hashcon is 
concurrently running with  resizing of $s$. So, a double check at  
\linref{hashcon-conv-doublecheck} is required to test whether $s$ is \nul 
between Lines \ref{lin:hashcon-conv-firstread} and \ref{lin:hashcon-conv-pred}. 
Then, we re-read that bucket of $t$ (\linref{hashcon-conv-secondread1} or 
\ref{lin:hashcon-conv-secondread2}), which must be initialized before $s$ 
becomes \nul, and then we perform the search in that bucket. If $b$ is not 
\nul, then we simply return the presence of the corresponding \vnode in the 
bucket $b$. Note that, at any point in time there are at most two \hnodes: only 
one when no resizing happens and another to support resizing -- halving or 
doubling -- of the hash table.

\subsection{The Edge Operations}
\label{subsec:nbk-eops}
The \nbk graph edge operations -- $\adde$, $\reme$, and $\cone$ -- are presented in \figref{nbk-graph-methods}. Before describing these operations, we detail the implementation of \find method, which is used by them. It is shown in \figref{nbk-bst-method1}. The method \find, at Lines \ref{findstart} to \ref{findend}, tries to locate the position of the key by traversing down the \elist of a \vnode. It returns the position in $pe$ and $ce$, and their corresponding \oper values in $peOp$ and $ceOp$ respectively. The result of the method \find can be one of the four values: (1) \found: if the key is present in the tree, (2) \notfoundl: if the key is not in the tree but might have been placed at the left child of $ce$ if it was added by some other threads, (3) \notfoundr: similar to \notfoundl but for the right child of $ce$, and (4) \abort: if the search in a subtree  is unable to return a usable result. 

 A $\adde(\vkey_1,\vkey_2|w)$ operation, at Lines \ref{addestart} to 
 \ref{addeend}, begins by validating the presence of $\vkey_1$ and $\vkey_2$ in 
 the \vlist. If the validations fails, it returns \eop{\fal}{$\infty$} 
 (\linref{adde:validation1}). Once the validation succeeds, \adde operation 
 invokes \find method in the \elist of the vertex with key $\vkey_1$ to 
 locate the position of the key $\vkey_2$. The position is returned in the 
 variables $pe$ and $ce$, and their corresponding \oper values are stored in 
 the $peOp$ and $ceOp$ respectively. On that, \adde checks whether 
 an \enode with the key $\vkey_2$ is 
 present. If it is present containing the same edge weight value \eweight, it 
 implies that an edge with the exact same weight is already present, therefore 
 \adde returns \eop{\fal}{$\infty$} (\linref{adde-found}). However, if it is 
 present with a different edge weight, say $z$, \adde updates 
 $ce$'s old weight $z$ to the new weight 
 $w$ and returns \eop{\tru}{$z$} (\linref{adde-update}). We update the 
 edge-weight using a \CAS to ensure the correct return in case there were 
 multiple concurrent \adde operations trying to update the same edge. Notice 
 that, here we are \textit{not} 
 freezing the \enode in anyway while updating its weight. The \lbty 
 is still ensured, which we discuss in the next section.

 If the key $\vkey_2$ is not present in the tree, a new 
 \enode and a \childcasop object are created. Then using \CAS the object is 
 inserted logically into $ce'$s \oper field (\linref{adde-casop}). If the \CAS 
 succeeds, it implies that $ce'$s \oper field hadn't been modified since the 
 first read. Which in turn indicates that all other fields of $ce$ were also 
 not changed by any other concurrent thread. Hence, the \CAS on one of the 
 $ce'$s child pointer should not fail. Thereafter, using a call to 
 \helpchildcas method the new \enode $ne$ is physically added 
 to the tree. This can be done by any thread that sees the ongoing operation in 
 $ce'$s \oper field.

A $\reme(\vkey_1,\vkey_2)$ operation, at Lines \ref{remestart} to 
\ref{remeend}, similarly begins by validating the presence of $\vkey_1$ and 
$\vkey_2$ in the \vlist. If the validation fails, it returns 
\eop{\fal}{$\infty$}. Once the validation succeeds, it invokes \find method in 
the \elist of the vertex having key $\vkey_1$ to locate the position of the key 
$\vkey_2$. If the key is not present it returns \eop{\fal}{$\infty$}. If the 
key is present, one of the two paths is followed. The first path at Lines 
\ref{lin:reme-1stpath-start} to \ref{lin:reme-1stpath-end} is followed if the 
node has less than two children. In case the node has 
both its children present a second path at Lines 
\ref{lin:reme-2ndpath-start} to \ref{lin:reme-2ndpath-end} is followed. The 
first path is relatively simpler to handle, as single \CAS 
instruction is used to mark the node from the state \none to \marked at this 
point the node is considered as logically deleted from the tree. After 
a successful \CAS, a \helpmarked method is invoked to perform the physical 
deletion. It uses a \childcasop to replace $pe'$s child pointer to $ce'$s with 
either a pointer to $ce'$s only child pointer, or a \nul pointer if $ce$ is a 
leaf node. 

 The second path is more difficult to handle, as the node has both the 
 children. Firstly, \find method only locates the children but an extra \find 
 (\linref{reme-2ndpath-start}) method is invoked to locate the node with the 
 next largest key. If the \find method returns \abort, which indicates that 
 $ce'$s \oper field was modified after the first search, so the entire \reme 
 operation is restarted. After a successful search, a \relocateop object  
 $replace$ is 
 created (\linref{reme-create-replaceop}) to replace  $ce'$s key $\vkey_2$ with 
 the node returned. This operation added to $replace'$s \oper 
 field safeguards it against a concurrent deletion while the \reme operation is 
 running by virtue of the use of a \CAS (\linref{reme-cas-replaceop}). Then 
 \helprelocate method is invoked to insert \relocateop into the node with 
 $\vkey_2'$s \oper field. This is done using a \CAS 
 (\linref{helprelocate-cas1}), after a successful \CAS the 
 node is considered as logically removed from the tree. Until the result of the 
 operation is known the initial state is set to \ongoing. If any other thread 
 either sees that the operation is completed by way of performing all the 
 required \CAS executions or takes steps to perform those \CAS operations  
 itself, it will set the operation state from \ongoing to \successful 
 (\linref{helprelocate-oning-succes}), using a \CAS. If it has seen other 
 value, it sets the operation state from \ongoing to \failed 
 (\linref{helprelocate-oning-fail}). After the successful state change, a \CAS 
 is used to update the key to new value and a second \CAS is used to delete the 
 ongoing \relocateop from the same node. Then next part of the \helprelocate 
 method performs cleanup on $replace$ by either marking it if the relocation 
 was successful or clearing its \oper field if it has failed. If the operation 
 is successful and $ce$ is marked, \helpmarked method is invoked to 
 excise $ce$ from the tree. At the end \reme returns 
 \eop{\tru}{$ce.\eweight$}
 
 Similar to \adde and \reme, a $\cone(\vkey_1,\vkey_2)$ operation, at Lines 
 \ref{conestart} to \ref{coneend}, begins by validating the presence of 
 $\vkey_1$ and $\vkey_2$ in the \vlist. If the validation fail, it returns 
 \eop{\fal}{$\infty$}. Once the validation succeeds, it invokes \find method in 
 the \elist of the vertex with key $\vkey_1$ to locate the position of the 
 key $\vkey_2$. If it finds $\vkey_2$, it checks if both the vertices are 
 not marked and also the $ceOp$ not marked; on ensuring that it returns 
 \eop{\tru}{$ce.\eweight$}, otherwise, it returns \eop{\fal}{$\infty$}.

\input{code/hash-code1.tex}

\input{code/bst-code1.tex}

\input{code/bst-code2.tex}

%% file: code/hash-bst-graph-code1.tex
\begin{figure*}[!thp]
\begin{multicols}{2}
\begin{algorithmic}[1]
\renewcommand{\algorithmicprocedure}{\textbf{Operation}}	\scriptsize
	\Procedure{ $\addv(\vkey$)}{}\label{addvstart}
	\State{return $\hashadd(\vkey);$}
	\EndProcedure\label{addvend}
	\algstore{addv}
\end{algorithmic}	
\hrule
    \begin{algorithmic}[1]
    \renewcommand{\algorithmicprocedure}{\textbf{Operation}}	
    \algrestore{addv}
    \scriptsize
	\Procedure{ $\remv$($\vkey$)}{}\label{remvstart}
	\State{return $\hashrem(\vkey);$}
	\EndProcedure\label{remvend}
	\algstore{remv}
\end{algorithmic}
\hrule
\begin{algorithmic}[1]
\renewcommand{\algorithmicprocedure}{\textbf{Operation}}		\algrestore{remv}
    \scriptsize
	\Procedure{ $\conv$($\vkey$)}{}\label{convstart}
    \State{$\langle st, v \rangle$ $\gets$ $\hashcon(\vkey);$}
	\If{($st$ = \tru)}
	\State return {$\langle \tru, v \rangle$;}
	\Else
	\State return {$\langle \fal, \nul \rangle$;}
	\EndIf
	\EndProcedure\label{convend}
	\algstore{conv}
\end{algorithmic}	
\hrule
\begin{algorithmic}[1]
\renewcommand{\algorithmicprocedure}{\textbf{Operation}}		\algrestore{conv}
    \scriptsize
	\Procedure{ $\cone$($\vkey_1, \vkey_2$)}{}\label{conestart}
    	\State {$\langle u,v,st \rangle$ $\gets$ $\convplus(\vkey_1,\vkey_2);$} \label{lin:convplus-uv}
		 \If{($st$ = \fal)}
		\State { return \eop{\fal}{$\infty$};}
        \EndIf
        \State{$\langle$$st$, $e \rangle$ $\gets$ $\bstcon$($\vkey_2$, $v.\enext$);}
	    \If{($st$ $=$ \found $\wedge$ $\neg$ $\hashcon(\vkey_1$) $\wedge$ $\neg$ $\hashcon(\vkey_2$))} \label{lin:cone-check}
	    \State {$z$ $\gets$ e.\eweight;}
        \State {return \eop{\tru}{$z$};}
        \Else \hspace{2mm} {return \eop{\fal}{$\infty$};}
        \EndIf
	    \EndProcedure\label{coneend}
		\algstore{cone}
\end{algorithmic}
\hrule
\begin{algorithmic}[1]
\renewcommand{\algorithmicprocedure}{\textbf{Method}}	\scriptsize
\algrestore{cone}
	\Procedure{ $\convplus(\vkey_1, \vkey_2$)}{}\label{convplusstart}
	\State{$\langle$$st1$,$u$$\rangle$ $\gets$ $\hashcon$($\vkey_1$);}\label{lin:convplus-conv1} \linecomment{//Modified \conv, returns status along with ref}
	\State{$\langle$$st2$,$v$$\rangle$$\gets$ $\hashcon$($\vkey_2$);}\label{lin:convplus-conv2}
	\If{($st1$ = \tru $\wedge$ $st2$ = \tru)}
    \State{return $\langle u, v, \tru \rangle $;}
    \Else \hspace{2mm}{return $\langle u, u, \fal \rangle $;}
    \EndIf
	\EndProcedure\label{convplusend}
	\algstore{convplus}
\end{algorithmic}
\hrule
   \begin{algorithmic}[1]
\renewcommand{\algorithmicprocedure}{\textbf{Operation}}	\scriptsize
\algrestore{convplus}
		\Procedure{ $\adde$($\vkey_1, \vkey_2| w$)}{}\label{addestart}
		\State {$\langle$ $u$, $v$, st $\rangle$ $\gets$ $\convplus$($\vkey_1, \vkey_2);$}  \label{lin:adde-convplus}
		 \If{($st$ = \fal)} { return \eop{\fal}{$\infty$};} \label{lin:adde:validation1}
        
        \EndIf
        \While{(\tru)}
	 \If{( $\isMarked(u$) $\vee$  $\isMarked(v$))}	\label{lin:adde-check-uv}
        \State { return \eop{\fal}{$\infty$};}
        \EndIf
          \State{$st$ $\gets$ $\find$($\vkey_2$, $pe$, $peOp$, $ce$ , $ceOp$, $v.\enext$);}\label{lin:wtadde-loceplus}
	    
	    \If{($\getflag$($pe$) $=$ \marked )} 
	    \State{\cntu;}
	    \EndIf
		\If{($st$ = \found )}
		\If{($ce.\eweight$ $=$ $w$)}  { return \eop{\fal}{$w$};} \label{lin:adde-found}
		\Else
		\State{$z$ $\gets$ $ce$.\eweight;} \label{lin:adde-update}
		\State{$\CAS$($ce$.\eweight,~$z$,~$w$);} \label{lin:adde-updtw}
		\State{$u$.\ecount.\fadd(1);}\label{lin:faa-updt}
		\State {return \eop{\tru}{$z$};}
		\EndIf
		\Else
		\State{$ne$ $\gets$ \createe($\vkey_2$, $w$);}
		\State{$ne.\pointv$ $\gets$ $v$;}
		\State{Boolean \isleft $\gets$ ($st$ = \notfoundl);}
        \State{\enode}  {$old$ $\gets$ \isleft $?$ $ce.\lft$ $:$ $ce.\rght$;}
        \State{$casOp$ $\gets$ new $\childcasop(\isleft,old,ne)$;}
        \If{($\CAS$($ce.\oper$,$ceOp$,$\flag$($casOp$,\childcas)))}\label{lin:adde-casop}
        \State{u.\ecount.\fadd(1);}\label{lin:faa-ade}
        \State{$\helpchildcas(casOp, ce)$;}
        \State { return \eop{\tru}{$\infty$};}
        \EndIf
		\EndIf
		\EndWhile
        \EndProcedure\label{addeend}
		\algstore{adde}
\end{algorithmic}  
\hrule


	\begin{algorithmic}[1]
\renewcommand{\algorithmicprocedure}{\textbf{Operation}}		\algrestore{adde}
\scriptsize
		\Procedure{ $\reme$($\vkey_1$, $\vkey_2$)}{}\label{remestart}
    		\State {$\langle u, v, st \rangle$ $\gets$ $\convplus$($\vkey_1, \vkey_2$);}  \label{lin:reme-validate-u}
		 \If {($st$ = \fal)} 
	    \State { return \eop{\fal}{$\infty$};}
       \EndIf
   	\While{(\tru)}
       	\If{($\isMarked(u$) $\vee$ $\isMarked(v$))}\label{lin:reme-check-uv}
        \State { return \eop{\fal}{$\infty$};}
        \EndIf
        
    \State{$st1$ $\gets$ $\find$($\vkey_2$, $pe$, $peOp$, $ce$ , $ceOp$, $v.\enext$);}
    \If{($st1$ $\neq$ \found)}
     \State { return \eop{\fal}{$\infty$};}
    \EndIf
    \If{($\isnull$($curr.\rght$) $\vee$ $\isnull$( $ce$.\lft)))}
    \If{($\CAS$($ce.\oper$, $ceOp$, $\flag(ceOp, \marked)$))} \label{lin:reme-1stpath-start}
    \State{u.\ecount.\fadd(1);} \label{lin:faa-reme1}
    \State{$\helpmarked(pe, peOp, ce);$}
    \State {$z$ $\gets$ ce.\eweight;}
   \State{\brk;} \label{lin:reme-1stpath-end}
    \EndIf
    \Else
   \State{$st2$ $\gets$ $\find$($\vkey_2$, $pe$, $peOp$, $ce$ , $ceOp$, $v.\enext$);}\label{lin:reme-2ndpath-start}
     \If{(($st2$ $=$ \abort) $\vee$ ($ce.\oper$ $\neq$ $ceOp$)}
     \State{\cntu;}
     \EndIf
     \State{$relocOp$ $\gets$ new $\relocateop$($ce$, $ceOp$, $\vkey_2$, $replace.\vkey_2)$;} \label{lin:reme-create-replaceop}
     \If{($\CAS$($replace.\oper$, $replaceOp$, $\flag$ ($relocOp$, \relocate))} \label{lin:reme-cas-replaceop}
     \State{u.\ecount.\fadd(1);} \label{lin:faa-reme2}
     \If{($\helprelocate$($relocOp$,$pe$,$peOp$,$replace$))}
    \State {$z$ $\gets$ ce.\eweight;}
     \State{\brk;} \label{lin:reme-2ndpath-end}
     \EndIf
     \EndIf
    \EndIf
    \EndWhile
    
		\State { return \eop{\tru}{$z$};}
	    \EndProcedure\label{remeend}
        \algstore{reme}
\end{algorithmic}
\end{multicols}
	\caption{Pseudocodes of \addv, \remv, \conv, \adde, \reme, \cone and \convplus}\label{fig:nbk-graph-methods}
\end{figure*}

%% file: code/hash-code1.tex
\begin{figure*}[!thb]
\captionsetup{font=scriptsize}
\begin{multicols}{2}
\begin{scriptsize}

		\begin{tabbing}
			\hspace{0.1in} \=
			\hspace{0.1in} \= 
			\hspace{0.1in} \= 
			\hspace{0.1in} \=  
			\hspace{0.1in} \= \\
			\> {\bf struct \fsetnode \{} \\
			\> \> {int \set;} // Set of integer \\
		    \> \> {boolean \ok;} // check for the \set is mutable or not\\
		    \> \}\\
		  \> {\bf struct \fset \{} \\
			\> \> {$\fsetnode$ $\node$;} \\
		    \> \}\\
		  
		  	\> {\bf struct \fsetop \{} \\
			\> \> {int \optype;} // operation type (\add or \rem) \\
		    \> \> {int \key;} // key to insert or delete\\
		     \> \> {boolean \resp;} // holds the return value\\
		    \> \}\\  
		    
		    \> {\bf struct \hnode \{} \\
			\> \> {\fset} {\buckets;} // an array(or list) of \fset \\
		    \> \> {int \Hsize;} // array(or list) length\\
		    \> \> {\hnode} {\hpred;} // pointer that points to \\
		   \>\>\>\>\> a predecessor \hnode\\
		    \> \}\\  
		\end{tabbing}
		\vspace{-0.2in}
	\end{scriptsize}
  \hrule
\begin{algorithmic}[1]
\renewcommand{\algorithmicprocedure}{\textbf{Method}}	\scriptsize
\algrestore{reme}
\Procedure{ $\getresponse$($op$)}{}\label{getresponsestart}
\State {return $op.\resp$;}
\EndProcedure\label{getresponseend}
\algstore{getresponse}
\end{algorithmic}
\hrule
    \begin{algorithmic}
\renewcommand{\algorithmicprocedure}{\textbf{Method}}	\scriptsize
\algrestore{getresponse}
		\Procedure{ $\hasmember$($b, k$)}{}\label{hasmemberstart}
        \State{$o$ $\gets$ $b.\node$;} // local copy of b  \label{lin:hasmember-read-bkt}
		\State {return $k \in o.\set$;}
        \EndProcedure\label{hasmemberend}
		\algstore{hasmember}
\end{algorithmic}
\hrule
\begin{algorithmic}
\renewcommand{\algorithmicprocedure}{\textbf{Method}}
    \scriptsize
    \algrestore{hasmember}
	\Procedure{ $\invoke$($b$, $op$)}{}\label{invokestart}
	\State{$o$ $\gets$ $b.\node$;} // local copy of b 
	\While{($o.\ok$)}
	\If{($op.\optype$ $=$ \add)}
	\State{$\resp$ $\gets$ $op.\key$ $\notin$ $o.\set$;}
	\State{$\set$ $\gets$ $o.\set$ $\cup$ $\{ op.\key \};$}
	\Else
	\If{($op.\optype$ $=$ \rem)}
	\State{$\resp$ $\gets$ $op.\key$ $\in$ $o.\set$;}
	\State{$\set$ $\gets$ $o.\set$ $\setminus$ $\{ op.\key \};$}
	\EndIf
	\EndIf
	\State{$n$ $\gets$ new $\fsetnode$(\set, \tru);}
    \If{(\CAS(b.\node,o,n));} 
    \State{$op.\resp$ $\gets$ $\resp$;}
    \State{return \tru;}
    \EndIf
    \State{ $o$ $\gets$ $b.\node$;}
    \EndWhile
    \State{return \fal;}
	\EndProcedure\label{invokeend}
	\algstore{invoke}
\end{algorithmic}	
 \hrule 

\begin{algorithmic}
\renewcommand{\algorithmicprocedure}{\textbf{Method}}	\scriptsize
\algrestore{invoke}
		\Procedure{ $\freeze$($b$)}{}\label{freezestart}
		\State{$o$ $\gets$ $b.\node$;} // local copy of b 
		\While{($o.\ok)$)}
        \State{$n$ $\gets$ new $\fsetnode$($o$.\set, \fal);}
        \If{(\CAS(b.\node,o,n));} \label{lin:freez-false-cas}
        \State{\brk;}
        \EndIf
    \State{ $o$ $\gets$ $b.\node$;}
    \EndWhile
    \State{return $o.\set$}
        \EndProcedure\label{freezeend}
		\algstore{freeze}
\end{algorithmic}
\hrule
    \begin{algorithmic}
    \renewcommand{\algorithmicprocedure}{\textbf{Operation}}	
    \algrestore{freeze}
    \scriptsize
	\Procedure{ $\hashadd$($\key$)}{}\label{addstart}
	\State{$\resp$ $\gets$ $\apply$($\add$, $\key$);} ]\label{lin:addv-apply}
	\If{($heuristic\text{-}policy$)}
    \State{\resize(\tru);}
    \EndIf
    \State return {\resp;}
	\EndProcedure\label{addend}
	\algstore{add}
\end{algorithmic}
\hrule
    \begin{algorithmic}[1]
    \renewcommand{\algorithmicprocedure}{\textbf{Operation}}	
    \algrestore{add}
    \scriptsize
	\Procedure{ $\hashrem$($\key$)}{}\label{remstart}
	\State{$\resp$ $\gets$ $\apply$($\rem$, $\key$);} \label{lin:remv-apply}
	\If{($heuristic\text{-}policy$)}
    \State{\resize(\fal);}
    \EndIf
	\State return {\resp;}
	\EndProcedure\label{remend}
	\algstore{rem}
\end{algorithmic}
\hrule

  \begin{algorithmic}[1]
    \renewcommand{\algorithmicprocedure}{\textbf{Operation}}	
    \algrestore{rem}
    \scriptsize
	\Procedure{ $\hashcon$($\key$)}{}\label{constart}
	\State{$t \gets \head;$}
	\State{$b \gets$ $t.\buckets$[\key \Mod $t$.\Hsize];} \label{lin:hashcon-conv-firstread}
	\If{($b$ $=$ $\nul$)} 
	\State{$s$ $\gets$ $t.\hpred$;} \label{lin:hashcon-conv-pred}
	\If{($s$ $\neq$ $\nul$ )} \label{lin:hashcon-conv-doublecheck}
	\State{$b \gets$ $s.\buckets$(\key \Mod $s$.\Hsize);}\label{lin:hashcon-conv-secondread1}
	\Else
    \State{$b \gets$ $t.\buckets$(\key \Mod $t$.\Hsize);}\label{lin:hashcon-conv-secondread2}
    \EndIf
    \EndIf
    \State{ return {\hasmember($b$,\key)};}
	\EndProcedure\label{conend}
	\algstore{con}
\end{algorithmic}
\hrule
  \begin{algorithmic}[1]
\renewcommand{\algorithmicprocedure}{\textbf{Method}}		\algrestore{con}
    \scriptsize
	\Procedure{ $\resize$($\grow$)}{}\label{growstart}
    	\State{$t \gets \head;$}
    	\If{($t.\Hsize$ $>$ $1$ $\vee$ $\grow$ $=$ \tru)} 
    	\For{($i$ $\gets$ $0$ to $t.\Hsize$-$1$)}
    	\State{$\initbucket$(t,i);}
    	\EndFor
    	\State{$t.\hpred$ $\gets$ $\nul$;}
    	\State{\Hsize $\gets$ \grow $?$ $t.\Hsize$ $\star$ $2$ $:$ $t.\Hsize$$/$$2$;}
    	\State{\buckets $\gets$ new $\fset$[\Hsize];}
    	\State{$t'$ $\gets$ new $\hnode$(\buckets, \Hsize, $t$);}
    	\State{$\CAS$(\head, $t$, $t'$);}\label{lin:resize-cas}
    	\EndIf
	    \EndProcedure\label{growend}
		\algstore{resize}
\end{algorithmic}  
\hrule
	\begin{algorithmic}[1]
\renewcommand{\algorithmicprocedure}{\textbf{Method}}		\algrestore{resize}
    \scriptsize
	\Procedure{ $\apply$($\optype$, $\key$)}{}\label{applystart}
	\State {$op$ $\gets$ new $\fsetop(\optype,\key, \fal, -)$;}
	\While{(\tru)}
	\State{$t \gets \head;$}
	\State{$b \gets$ $t.\buckets$[\key \Mod $t$.\Hsize];}
	\If{($b$ $=$ $\nul$)} 
	\State{$b$ $\gets$ $\initbucket$($t$,$\key$.\Mod $t.\Hsize$);} \label{lin:apply-initbcucket}
	\EndIf
	\If{($\invoke(b,op)$)} \label{lin:apply-invoke}
	\State return {\getresponse($op$);}
	\EndIf
	\EndWhile
	\EndProcedure \label{applyend}
	\algstore{apply}
\end{algorithmic}	 
\hrule	
	\begin{algorithmic}[1]
\renewcommand{\algorithmicprocedure}{\textbf{Method}}		\algrestore{apply}
\scriptsize
		\Procedure{ $\initbucket$($t$, $\key$)}{}\label{initbucketstart}
    	\State{$b \gets t.\buckets[\key];$}
    	\State{$s \gets t.\hpred;$}
    	\If{($b$ $=$ $\nul$ $\wedge$ $s$ $\neq$ $\nul$)} 
        \If{($t.\Hsize$ $=$ $s.\Hsize$)} \State{$m$ $\gets$ $s.\buckets$[$i$ \Mod $s.\Hsize$];}
        \State{\set $\gets$ $\freeze(m)$ $\cap$ \{ $x$ $\mid$ $x$ \Mod $t.\Hsize$ $=$ $i$ \};}
        \Else
         \State{$m$ $\gets$ $s.\buckets$[$i$];}
         \State{$m'$ $\gets$ $s.\buckets$[$i$ $+$ $s.\Hsize$];}
        \State{\set $\gets$ $\freeze(m)$ $\cup$ $\freeze(m')$;}
    	\EndIf
    	\State{$b'$ $\gets$ new $\fset(\set, \tru)$;}
    	\State{$\CAS$($t.\buckets[i]$, $\nul$, $b'$);}
    	\EndIf
    	\State{ return $t.\buckets[i]$;}
	    \EndProcedure\label{initbucketend}
       \algstore{initbucket}
\end{algorithmic}
\end{multicols}
	\caption{Strucure of \fset, \fsetop and \hnode. Pseudocodes of \invoke, \freeze, \add, \rem, \con, \resize, \apply and \initbucket methods based on dynamic sized \nbk hash table\cite{Liu+:LFHash:PODC:2014}.}\label{fig:nbk-hash-methods}
\end{figure*}

%% file: code/bst-code1.tex
\begin{figure*}[!thb]
\captionsetup{font=scriptsize}
\begin{multicols}{2}
\begin{scriptsize}
		\begin{tabbing}
			\hspace{0.1in} \=
			\hspace{0.1in} \= 
			\hspace{0.1in} \= 
			\hspace{0.1in} \=  
			\hspace{0.1in} \= \\
			\> {\bf struct \Node \{} \\
			\> \> {int \key;} //  \\
		    \> \> {\operation} {\oper;} // \\
		    \> \> {\Node} {\lft;} // \\
		    \> \> {\Node} {\rght;}// \\
		    \> \}\\
		  
		  	\> {\bf struct \relocateop \{} \\
			\> \> {int \statee;} // initial \statee=\ongoing  \\
		    \> \> {\Node} {\dest;} // \\
		    \> \> {\operation} {\destop;} // \\
		     \> \> {int \removekey;} // \\
		    \> \> {int \replacekey;} // \\
		    \> \}\\  
		    
		    \> {\bf struct \childcasop \{} \\
			\> \> {boolean \isleft;} //  \\
		    \> \> {\Node} {\expected;} // \\
		    \> \> {\Node} {\updt;} // \\
		    \> \}\\  
		\end{tabbing}
		\vspace{-0.2in}
	\end{scriptsize}

    \begin{algorithmic}[1]
    \renewcommand{\algorithmicprocedure}{\textbf{Operation}}	
    \algrestore{initbucket}
    \scriptsize
	\Procedure{ $\add$($\key$)}{}\label{bstaddstart}
    \State{\Node{} $pred$, $curr$, $newNode$;}
    \State{\operation $predOp$, $currOp$, $casOp$;}
    \State{int $result$;}
    \While{(\tru)}
    \State{$result$ $\gets$ $\find$(\key, $pred$, $predOp$, $curr$, $currOp$, $root$);}
    \If{($result$ = \found) }
    \State{return \fal;}
    \EndIf
    \State{$newNode$ $\gets$ new $\Node(\key)$;}
    \State{Boolean \isleft $\gets$ ($result$ = \notfoundl);}
    \State{\Node}  {$old$ $\gets$ \isleft $?$ $curr.\lft$ $:$ $curr.\rght$;}
    \State{$casOp$ $\gets$ new $\childcasop(\isleft, old, newNode)$;}
    \If{($\CAS$($curr.\oper$, $currOp$, $\flag$($casOp$,\childcas)))}
    \State{$\helpchildcas(casOp, curr)$;}
    \State{return \tru;}
    \EndIf
    \EndWhile
	\EndProcedure\label{bstaddend}
	\algstore{add}
\end{algorithmic}
\hrule
    \begin{algorithmic}[1]
    \renewcommand{\algorithmicprocedure}{\textbf{Operation}}	
    \algrestore{add}
    \scriptsize
	\Procedure{ $\rem$($\key$)}{}\label{bstremstart}
	\State{\Node $pred$, $curr$, $replace$;}
    \State{\operation $predOp$, $currOp$, $replaceOp$, $relocOp$;}
    \While{(\tru)}
    \If{($\find$(\key, $pred$, $predOp$, $curr$, $currOp$, $root$) $\neq$ \found)}
    \State{return \fal;}
    \EndIf
    \If{($\isnull$($curr.\rght$ $\vee$ $\isnull$( $curr$.\lft)))}
    \If{($\CAS$($curr.\oper$, $currOp$, $\flag(currOp, \marked)$))}
    \State{$\helpmarked(pred, predOp, curr);$}
    \State{return \tru;}
    \EndIf
    \Else
     \If{(($\find$(\key, $pred$, $predOp$, $replace$, $replaceOp$, $curr$) $=$ \abort) $\vee$ ($curr.\oper$ $\neq$ $currOp$)}
     \State{\cntu;}
     \EndIf
     \State{$relocOp$ $\gets$ new $\relocateop(curr, currOp, \key, replace.\key)$;}
     \If{($\CAS$($replace.\oper$, $replaceOp$, $\flag(relocOp$, \relocate))}
     \If{($\helprelocate$($relocOp$, $pred$, $predOp$, $replace$))}
     \State{return \tru;}
     \EndIf
     \EndIf
    \EndIf
    \EndWhile
	\EndProcedure\label{bstremend}
	\algstore{rem}
\end{algorithmic}
\hrule
    
  \begin{algorithmic}[1]
    \renewcommand{\algorithmicprocedure}{\textbf{Operation}}	
    \algrestore{rem}
    \scriptsize
	\Procedure{ $\con$($\key$)}{}\label{bstconstart}
	\State{\Node ~$pred$, $curr$;}
    \State{\operation $predOp$, $currOp$;}
    \If{($\find$(\key, $pred$, $predOp$, $curr$, $currOp$, $root$) $=$ \found)}
    \State{return \tru;}
    \Else
     \State{return \fal;}
    \EndIf
	\EndProcedure\label{bstconend}
	\algstore{con}
\end{algorithmic}
\hrule
  \begin{algorithmic}[1]
\renewcommand{\algorithmicprocedure}{\textbf{Method}}		\algrestore{con}
    \scriptsize
	\Procedure{ $\find$(\key, $pred$, $predOp$, $curr$, $currOp$, $root$)}{}\label{findstart}
   \State{int $result$, $currKey$};
\State{\Node} {$next$, $lastRight$};
\State{\operation $lastRightOp$};
\State{$result$ $\gets$ \notfoundr}; \label{lin:find-retry}
\State{$curr$ $\gets$ $root$};
\State{$currOp$ $\gets$ $curr.\oper$;}
\If{($\getflag(currOp)$ $\neq$ \nul) }
\If{($root$ $=$ \Root) }
\State{$\helpchildcas$($\unflag$($currOp$), $curr$)};
\State{goto Line \ref{lin:find-retry};}
\Else 
\State{return \abort};
\EndIf
\EndIf
\State{$next$ $\gets$ $curr.\rght$};
\State{$lastRight$ $\gets$ $curr$;}
\State{$lastRightOp$ $\gets$ $currOp$;}
\While{($\neg$ $\isnull(next)$)}
\State{$pred$ $\gets$ $curr$;}
\State{$predOp$ $\gets$ $currOp$;}
\State{$curr$ $\gets$ $next$;}
\State{$currOp$ $\gets$ $curr.\oper$;}
\If{($\getflag$($currOp$) $\neq$ \nul)}
\State{$\help$($pred$, $predOp$, $curr$, $currOp$);}
\State{goto Line \ref{lin:find-retry};}
\EndIf
\State{$currKey$ $\gets$ $curr.\key$;} \label{lin:find-read-curr}
\If{(\key $<$ $currKey$)}
\State{$result$ $\gets$ \notfoundl;}
\State{$next$ $\gets$ $curr.\lft$;} \label{lin:find-read-curr-left}
\Else 
\If{($\key$ $>$ $currKey$)}
\State{$result$ $\gets$ \notfoundr;}
\State{$next$ $\gets$ $curr.\rght$;} \label{lin:find-read-curr-right}
\State{$lastRight$ $\gets$ $curr$;}
\State{$lastRightOp$ $\gets$ $currOp$;}
\Else 
\State{$result$ $\gets$ \found;}
\State{\brk;}
\EndIf
\EndIf
\EndWhile
\If{(($result$ $\neq$ \found) $\bigwedge$ ($lastRightOp$ $\neq$ $lastRight.\oper$)}
\State{goto Line \ref{lin:find-retry};}
\EndIf
\If{($curr.\oper$ $\neq$ $currOp$) }
\State{goto Line \ref{lin:find-retry};}
\EndIf
\State{ return $result$;}
	    \EndProcedure\label{findend}
		\algstore{find}
\end{algorithmic}  
\end{multicols}	
	\caption{Strucure of \Node, \relocateop and \childcasop. Pseudocodes of \add, \rem, \con and \find methods based on \nbk binary search tree\cite{Howley+:NbkBST:SPAA:2012}\label{fig:nbk-bst-method1}}
\end{figure*}

%% file: code/bst-code2.tex
\begin{figure*}[!thb]
\captionsetup{font=scriptsize}
\begin{multicols}{2}
  \begin{algorithmic}[1]
    \renewcommand{\algorithmicprocedure}{\textbf{Method}}	
    \algrestore{find}
    \scriptsize
	\Procedure{ $\helpchildcas$($\oper$, $dest$)}{}\label{helpchildcasstart}
	\State{\Node $address$ $\gets$ $\oper.\isleft$ $?$ $dest.\lft$ $:$ $dest.\rght$;}
	\State{$\CAS$($address$, $\oper.\expert$, $\oper.\updt$);}
	\State{$\CAS$($dest.\oper$, $\flag(\oper$,\caschild ),$\flag(\oper$,\none))}
	\EndProcedure\label{helpchildcasend}
	\algstore{helpchildcas}
\end{algorithmic}
\hrule
 \begin{algorithmic}[1]
    \renewcommand{\algorithmicprocedure}{\textbf{Method}}	
    \algrestore{helpchildcas}
    \scriptsize
	\Procedure{ $\helprelocate$($\oper$, $pred$, $predOp$, $curr$)}{}\label{helprelocatestart}
	\State{$seenState$ $\gets$ $\oper.\statee$;}
	\If{($seenState$ $=$ \ongoing)}
\State{\operation $seenOp$ $\gets$ $\VCAS$(\oper.$dest.\oper$, $\oper.destOp$, $\flag$(\oper, \relocate));} \label{lin:helprelocate-cas1}
 \If{(($seenOp$ = $\oper.destOp$) $\bigvee$ ($seenOp$ = $\flag$(\oper, \relocate)))}
\State{$\CAS$(\oper.\statee, \ongoing, \successful);} \label{lin:helprelocate-oning-succes}
\State{$seenState$ $\gets$ \successful;}
\Else
\State{$seenState$ $\gets$ $\VCAS$(\oper.\statee, \ongoing, \failed);}\label{lin:helprelocate-oning-fail}
\EndIf
\EndIf
\If{($seenState$ = \successful)}
\State{$\CAS$(\oper.\dest.\key, $removeKey$, $replaceKey$);}
\State{$\CAS$(\oper.\dest.\oper, $\flag$(\oper, \relocate), $\flag$(\oper, \none));}
\EndIf
\State{boolean $result$ $\gets$ ($seenState$ = \successful);}
\If{(\oper.\dest = $curr$)}
\State{return $result$;}
\EndIf
\State{$\CAS$($curr$.\oper, $\flag$(\oper, \relocate), $\flag$(\oper, $result$ $?$ \marked : \none));}
\If{($result$)}
\If{(\oper.\dest = $pred$)}
\State{$predOp$ $\gets$ $\flag$(\oper, \none);}
\EndIf
\State{\helpmarked($pred$, $predOp$, $curr$);}
\EndIf
\State{return $result$;}
	\EndProcedure\label{helprelocateend}
	\algstore{helprelocate}
\end{algorithmic}
\hrule
\begin{algorithmic}[1]
    \renewcommand{\algorithmicprocedure}{\textbf{Method}}	
    \algrestore{helprelocate}
    \scriptsize
	\Procedure{ $\helpmarked$($pred$, $predOp$, $curr$)}{}\label{helpmarkedstart}
	\State{\Node $newRef$;}
    \If{($\isnull$($curr.\lft$))}
    \If{($\isnull$($curr.\rght$))}
    \State{$newRef$ $\gets$ $\setnull$($curr$);}
    \Else
    \State{$newRef$ $\gets$ $curr.\rght$;}
    \EndIf
    \Else
    \State{$newRef$ $\gets$ $curr.\lft$;}
    \EndIf
    \State{\operation $casOp$ $\gets$ new $\childcasop$($curr$ $\gets$ $pred.\lft$,
$curr$, $newRef$);}
\If{($\CAS$($pred.\oper$, $predOp$, $\flag$($casOp$, CHILDCAS)))}
\State{$\helpchildcas$($casOp$, $pred$);}
\EndIf
	\EndProcedure\label{helpmarkedend}
	\algstore{helpmarked}
\end{algorithmic}
\hrule
\begin{algorithmic}[1]
    \renewcommand{\algorithmicprocedure}{\textbf{Method}}	
    \algrestore{helpmarked}
    \scriptsize
	\Procedure{ $\help$($pred$, $predOp$, $curr$, $currOp$)}{}\label{helpstart}
    \If{($\getflag$($currOp$) = \childcas)}
    \State{$\helpchildcas(\unflag(currOp),curr)$;}
    \Else
    \If{($\getflag$($currOp$) = \relocate)}
    \State{$\helprelocate(\unflag(currOp), pred, predOp, curr)$;}
    \Else
    \If{($\getflag$($currOp$) = \marked)}
    \State{$\helpmarked(pred, predOp, curr)$;}
    \EndIf
    \EndIf
    \EndIf
	\EndProcedure\label{helpend}
\end{algorithmic}
\end{multicols}	
	\caption{Pseudocodes of \helpchildcas, \helprelocate, \helpmarked and \help methods based on \nbk binary search tree\cite{Howley+:NbkBST:SPAA:2012}\label{fig:nbk-bst-method2}}
\end{figure*}

%% file: proof.tex
We argue that a linearizable \cite{Herlihy+:lbty:TPLS:1990} implementation 
maintains the data structure invariant. To prove linearizability, we specify 
the atomic events corresponding to the linearization points ($\lp$)  inside the 
execution interval of each of the operations.
\subsection{Linearizability}
Each operation implemented by the \ds are represented by their invocation and 
return steps. We show that it is possible to assign an atomic step as $\lp$ 
inside the execution interval of each operation. The vertex operations have 
their \lp{s} along the similar lines as that discussed in 
\cite{Liu+:LFHash:PODC:2014}. However, the edge operations include updating the 
weights of \enodes in addition to their addition and removal. Accordingly, we 
have more execution cases compared to a set implemented by a \nbk BST as 
implemented in \cite{Howley+:NbkBST:SPAA:2012}. The specification of \lp{s} of 
\getbfs, \getsp, and \getbc are closer to that of \textsc{GetPath} operation 
of \cite{Chatterjee+:NbGraph:ICDCN-19}.
\begin{theorem}\normalfont The ADT operations implemented by the \nbk graph algorithm are linearizable.
\end{theorem}
\begin{proof}
Based on the return values of the operations we discuss the \lp{s}.
\input{code/lps.tex}
From the above description, one can notice that the \lp{s} of each of the operations lie in 
the interval between their invocation and the return steps. We can observe that in any invocation of a $\addv$ or a $\remv$ operation, 
with key value $\vkey$, there is always a unique \fset object for an operation to be applied. If two 
threads $T_1$ and $T_2$ try to add or delete the same key $\vkey$, after 
indexing in the hash table using modular arithmetic and both call the \invoke method. Thus, both either hash to the same bucket $b$, or at least one of them gets mapped to an immutable bucket. This prevents multiple \invoke methods parallely add or 
remove the same key at different buckets, and thereby a possible event invalidating linearizability is avoided. The $\adde$ and $\reme$ operations are similar to \cite{Howley+:NbkBST:SPAA:2012} except the case when the edge-weight is 
updated. However, update of edge-weight does not interfere with the arrangement 
of nodes in the BST corresponding to the \elist. The non-update operations, 
\conv, \cone, \getbfs, \getsp, and \getbc do not modify the \ds. Thus, 
following from \cite{Liu+:LFHash:PODC:2014} and \cite{Howley+:NbkBST:SPAA:2012} 
we conclude that all \nbk graph operations maintain the invariant of the \ds  
across the \lp{s}. This completes the correctness proof.
\end{proof}
\subsection{Non-blocking Progress Guarantee}
\begin{theorem}\normalfont
The presented concurrent graph operations
	\begin{enumerate}[label=(\roman*)]
		\item If the set of keys is finite, the operations \conv and \cone are wait-free.\label{lflem1}
		\item The operation \getbfs, \getsp, and \getbc are \of. \label{lflem2}
		\item The operations \addv, \remv, \conv, \adde, \reme, and \cone are \lf. \label{lflem3}
	\end{enumerate}
\end{theorem}
\begin{proof}
If the set of keys is finite, the size of the concurrent graph has a fixed upper bound. This implies that there are only a finite number of \vnodes in each bucket. A search for a given \vnode having key $\vkey$ is either in the bucket $b$, or if $b$ is \nul, in the predecessor's buckets, so it terminates in a finite number of steps. Similarly, a \cone operation invokes the \find method and it terminates by traversing the tree until the key is found or a null node is reached.
A \getbfs operation or a \comparetree method never returns \tru with concurrent update operations, which impose the \textbf{While} loop (Line \ref{bfswhilescan1}) in \bfsscan method to not terminate. So, unless a non-faulty thread has taken the steps in isolation a \getbfs operation will never return.  A similar argument can be brought for a \getsp, and a \getbc operation. This shows the \ref{lflem2}. 
Whenever an insertion or a deletion operation is blocked by a concurrent delete operation by the way of a marked pointer, then that blocked operations is helped  to make a safe return. Generally, insertion and lookup operations do not need help by a concurrent operation. So, if any random concurrent execution consists of any concurrent \ds operation, then at least one operation finishes its execution in a finite number of steps taken be a non-faulty thread. Therefore,  the concurrent graph operations \addv, \remv, \conv, \adde, \reme, and \cone are \lf. This shows the \ref{lflem3}.
\end{proof}

%% file: code/lps.tex
\begin{enumerate}[leftmargin=5.5mm]
	\item $\addv(\vkey)$: We have two cases:
	\begin{enumerate}
		\item \tru: The key $\vkey$ was not present earlier, then the \lp is 
		at  \linref{apply-invoke}, the \invoke method that returns \tru where 
		it sets \oper.\done to \tru.
		\item \fal: The key $\vkey$ was already present, then the \lp is at  
		\linref{apply-invoke}, the \invoke method that returns \tru where it 
		sets \oper.\done to \tru.
	\end{enumerate}
	\item $\remv(\vkey)$: We have two cases:
	\begin{enumerate}
		\item \tru: If the key $\vkey$ was already present, then the \lp is at  
		\linref{apply-invoke}.
		\item \fal:  If the key $\vkey$ was not present earlier, then the \lp 
		is at \linref{apply-invoke}.
	\end{enumerate}
	\item $\conv(\vkey)$:  We have two cases:
	\begin{enumerate}
		\item \tru: The \lp is at  \linref{hashcon-conv-secondread1} or 
		\ref{lin:hashcon-conv-secondread2},  second read of the bucket of $t$ 
		or $s$, respectively.
		\item \fal: The \lp is at \linref{freez-false-cas} because $b$ must 
		have been made immutable by some \freeze method that sets $b.\ok$ to 
		\fal.
	\end{enumerate}
	
	\item $\adde(\vkey_1, \vkey_2|w)$: We have four cases:
	\begin{enumerate}
	    \item \eop{\tru}{$\infty$}: New edge has been added
	    \label{step:eadd}
		\begin{enumerate}
			\item No concurrent $\remv(\vkey_1)$ or $\remv(\vkey_2)$: The \lp 
			is the successful \CAS execution at the \linref{adde-casop}.
			\item With concurrent $\remv(\vkey_1)$ or $\remv(\vkey_2)$: The \lp 
			is just before the first remove's \lp. 
		\end{enumerate}
		
	    \item \eop{\tru}{$z$}: Edge has been updated 
	    \begin{enumerate}
			\item No concurrent $\remv(\vkey_1)$ or $\remv(\vkey_2)$ or 
			$\reme(\vkey_1, \vkey_2)$: The \lp is the atomic update of the 
			edge-weight to \eweight using a successful \CAS, at 
			\linref{adde-updtw}. 
			\item With concurrent $\remv(\vkey_1)$ or $\remv(\vkey_2)$ or 
			$\reme(\vkey_1, \vkey_2)$: The \lp is just before the first 
			remove's \lp. 
		\end{enumerate}	
	    \item \eop{\fal}{$w$}: Edge already present \begin{enumerate}
			\item No concurrent $\remv(\vkey_1)$ or $\remv(\vkey_2)$ or 
			$\reme(\vkey_1, \vkey_2)$: The \lp is the atomic read of the \enode 
			$e(\vkey_2)$ has occurred at \linref{find-read-curr} in the \find 
			method.
			\item With concurrent $\remv(\vkey_1)$ or $\remv(\vkey_2)$ or 
			$\reme(\vkey_1, \vkey_2)$: The \lp is just before the first 
			remove's \lp. 
		\end{enumerate}	
	    \item \eop{\fal}{$\infty$}: Either $\vkey_1$ or $\vkey_2$ or both are not present
		\begin{enumerate}
			\item No concurrent $\remv(\vkey_1)$ or $\remv(\vkey_2)$ or 
			$\reme(\vkey_1, \vkey_2)$: The \lp is the end of the search path 
			reached by reading a \nul pointer at \linref{find-read-curr-left} 
			or \ref{lin:find-read-curr-right} in the \find method.
			\item At the time of invocation of $\adde(\vkey_1, \vkey_2|w)$ if 
			both vertices $\vkey_1$ and $\vkey_2$ were in the \vlist and a 
			concurrent \remv removed $\vkey_1$ or $\vkey_2$ or both then the 
			\lp is the just after the \lp of the earlier \remv.
			\item At the time of invocation of $\adde(\vkey_1, \vkey_2|w)$ if 
			both vertices $\vkey_1$ and $\vkey_2$ were not present in the 
			\vlist, then the \lp is the invocation point itself.
		\end{enumerate}
	\end{enumerate}			
	\item $\reme(\vkey_1, \vkey_2)$:  We have two cases:
	\begin{enumerate}
		\item \eop{\tru}{$w$}: Edge already present \begin{enumerate}
			\item No concurrent $\remv(\vkey_1)$ or $\remv(\vkey_2)$ or $\reme(\vkey_1, \vkey_2)$: We have two cases
			\begin{enumerate}
			\item With less than two children: The \lp is at 
			\linref{reme-1stpath-start}, when the \enode is set as marked. 
			\item With two children: It needs two {\CAS}s to succeed, one on 
			the node having the key to be deleted, and second which has the 
			next largest key. So a successful \reme's \lp is at 
			\linref{helprelocate-cas1}, when the \relocateop is installed by 
			the \helprelocate method.
			\end{enumerate}	
			\item With concurrent $\remv(\vkey_1)$ or $\remv(\vkey_2)$ or 
			$\reme(\vkey_1, \vkey_2)$: The \lp is just before the first 
			remove's \lp. 
		\end{enumerate}	
	    \item \eop{\fal}{$\infty$}: Either $\vkey_1$ or $\vkey_2$ or both are not present
		\begin{enumerate}
		    \item No concurrent $\remv(\vkey_1)$ or $\remv(\vkey_2)$ or 
		    $\reme(\vkey_1, \vkey_2)$: The \lp is the end of the search path 
		    reached by reading a \nul pointer at \linref{find-read-curr-left} 
		    or \ref{lin:find-read-curr-right} in the \find method.
			\item At the time of invocation of $\reme(\vkey_1, \vkey_2)$ if 
			both vertices $\vkey_1$ and $\vkey_2$ were in the \vlist and a 
			concurrent \remv removed $\vkey_1$ or $\vkey_2$ or both then the 
			\lp is the just after the \lp of the earlier \remv.
			\item At the time of invocation of $\reme(\vkey_1, \vkey_2)$ if 
			both vertices $\vkey_1$ and $\vkey_2$ were not present in the 
			\vlist, then the \lp is the invocation point itself.
		\end{enumerate}
	
	\end{enumerate}			
	
		\item $\cone(\vkey_1, \vkey_2)$: We have two cases:
	\begin{enumerate}
		\item \eop{\tru}{$w$}: Edge present \begin{enumerate}
			\item No concurrent $\remv(\vkey_1)$ or $\remv(\vkey_2)$ or 
			$\reme(\vkey_1, \vkey_2)$: The \lp is the atomic read of the key 
			that occurs at \linref{find-read-curr} in the method \find.
			\item With concurrent $\remv(\vkey_1)$ or $\remv(\vkey_2)$ or 
			$\reme(\vkey_1, \vkey_2)$: The \lp is just before the first 
			remove's \lp. 
		\end{enumerate}	
	    \item \eop{\fal}{$\infty$}: Either $\vkey_1$ or $\vkey_2$ or both or $e(\vkey_2)$ are not present
		\begin{enumerate}
		    \item No concurrent $\remv(\vkey_1)$ or $\remv(\vkey_2)$ or 
		    $\reme(\vkey_1, \vkey_2)$: The \lp is the end of the searching path 
		    reached by reading a \nul pointer at \linref{find-read-curr-left} 
		    or \ref{lin:find-read-curr-right} in the \find method.
		   	\item At the time of invocation of $\cone(\vkey_1, \vkey_2)$ if 
		   	both vertices $\vkey_1$ and $\vkey_2$ were in the \vlist and a 
		   	concurrent \remv removes $\vkey_1$ or $\vkey_2$ or both then the 
		   	\lp is just after the \lp of the earlier \remv.
			\item At the time of invocation of $\cone(\vkey_1, \vkey_2)$ if 
			both vertices $\vkey_1$ and $\vkey_2$ were not present in the 
			\vlist, then the \lp is the invocation point itself.
		\end{enumerate}
		
	\end{enumerate}

	\item $\getbfs(\vkey):$	Here, there are two cases:
	\begin{enumerate}
		\item \getbfs invoke the \bfsscan method: Assuming that \bfsscan 
		invokes $m$ (greater than equal to 2) \bfstclt procedures, it is the 
		last atomic read step of the $(m-1)^{st}$ \bfstclt call.
		\item \getbfs does not invoke the \bfsscan method: If a concurrent \remv operation $op$ removed $\vkey$. Then just after the LP of $op$. If $\vkey$ did not exist in the \vlist before the invocation then at the invocation of $\getbfs(\vkey)$.
		\label{step:getbfs-b}
    \end{enumerate}
\item $\getsp(\vkey):$	Similarly, there are two cases here as well:
	\begin{enumerate}
		\item \getsp invoke the \spscan method: Assuming that \spscan invokes 
		$m$ (greater than equal to 2) \sptclt procedures, it is the last atomic 
		read step of the $(m-1)^{st}$ \sptclt call.
		\item \getsp does not invoke the \spscan method: The \lp is the same as 
		the case \ref{step:getbfs-b} of \getbfs operation. 
    \end{enumerate}
\ignore{
\item $\getdia():$	Similar to the above two operations:
	\begin{enumerate}
		\item \getdia invoke the \diascan method: Assuming that \diascan 
		invokes $m$ (greater than equal to 2) \diametertclt procedures, it is 
		the last atomic read step of the $(m-1)^{st}$ \diametertclt call.
		\item \getdia does not invoke the \diascan method: If a concurrent 
		\remv operation $op$ removed $\vkey_i$, $\forall i\in V$, then just 
		after the LP of $op$. 
    \end{enumerate}
}    
\item $\getbc(\vkey):$	Here, there are two cases:
	\begin{enumerate}
		\item \getbc invoke the \bcscan method: Assuming that \bcscan invokes $m$ (greater than equal to 2) \bctclt procedures. Then it is the last atomic read step of the $(m-1)^{st}$ \bctclt call.
		\item \getbc does not invoke the \bcscan method: The \lp is the same as 
		the case \ref{step:getbfs-b} of \getbfs operation. 
\end{enumerate}
\end{enumerate}